\documentclass[a4paper,UKenglish]{lipics-v2016}
\usepackage{amsmath,amssymb,amsthm,graphicx,fixmath,tikz,latexsym,color,xspace}
\usepackage{microtype}

\graphicspath{{figures/}}

\usepackage[left]{lineno}

\theoremstyle{plain}
\newtheorem{observation}[theorem]{{Observation}}
\newtheorem{proposition}[theorem]{{Proposition}}

\newcommand{\Z}{\mathbb{Z}}
\newcommand{\R}{\mathbb{R}}
\newcommand{\CDO}{TOC}
\newcommand{\CDR}{TOC}
\newcommand{\CDS}{TOC}
\newcommand{\mat}{\mathsf}

\def\marrow{{\marginpar[\hfill$\longrightarrow$]{$\longleftarrow$}}}
\newif\ifComments
\Commentstrue

\usepackage{xcolor}
\ifComments
  \newcommand{\says}[3]{\textcolor{#1}{\textsc{#2 says:} \marrow\textsf{#3}}}
\else
  \newcommand{\says}[2]{\relax}
\fi

\newcommand{\mati}[1]{\says{red}{Mati}{#1}}

\title{High Dimensional Consistent Digital Segments\footnote{M.K. was partially supported in part by the ELC project (MEXT KAKENHI No.~12H00855 and 15H02665).}}
\titlerunning{High Dimensional Consistent Digital Line Segments} %

\author[1,2]{Man-Kwun Chiu}
\author[3]{Matias Korman}
\affil[1]{National Institute of Informatics (NII), Tokyo, Japan. \\
  \texttt{chiumk@nii.ac.jp}}
\affil[2]{JST, ERATO, Kawarabayashi Large Graph Project.}
\affil[3]{Tohoku University, Sendai, Japan.\\
\texttt{mati@dais.is.tohoku.ac.jp}}

\authorrunning{M.\,K. Chiu and M.\,Korman} %

\Copyright{M.\,K. Chiu and M.\,Korman}%

\subjclass{"I.3.5 Computational Geometry and Object Modeling", "I.4.1 Digitization and Image Capture"}

\keywords{Consistent Digital Line Segments, Digital Geometry, Computer Vision}%

\begin{document}
\maketitle

\begin{abstract}
We consider the problem of digitalizing Euclidean line segments from $\R^d$ to $\Z^d$. Christ {\em et al.} (DCG, 2012) showed how to construct a set of {\em consistent digital segment} (CDS) for $d=2$: a collection of segments connecting any two points in $\Z^2$ that satisfies the natural extension of the Euclidean axioms to $\Z^d$. In this paper we study the construction of CDSs in higher dimensions. 

We show that any total order can be used to create a set of {\em consistent digital rays} CDR in $\Z^d$ (a set of rays emanating from a fixed point $p$ that satisfies the extension of the Euclidean axioms). We fully characterize for which total orders the construction holds and study their Hausdorff distance, which in particular positively answers the question posed by Christ {\em et al.}.
\end{abstract}

\section{Introduction}
Computation in Ancient Greece was rigorously done with ruler and compass using the five axioms of Euclidean geometry. The study of these axioms has had a drastic influence in the development of mathematics. Indeed, the removal of one of them (the fifth one) created non-Euclidean geometries, which have had huge influence on science and technology.

Computers and digital data have nowadays replaced the ruler and compass methods of computation.  In order to have a rigorous system of geometric computation in digital world, it is desirable to establish a set of axioms similar to those Euclidean geometry, where we need to replace a line by a Manhattan path in the micro scale that in a macro scale can be seen as a straight line.

There have been several attempts to define digital segments in a two dimensional $n \times n$ grid. The two dimensional bounded space is the most popular case to consider given its many applications in computer vision and computer graphics. Solutions have been proposed from a robust finite-precision geometric computation point of view~\cite{DBLP:conf/focs/GreeneY86,DBLP:conf/iccS/Sugihara01}, snap rounding~\cite{DBLP:conf/compgeom/GoodrichGHT97}, and many more. 

A pioneering work by Michael Luby in 1987~\cite{luby} introduced an axiomatic approach of  the set of digital rays emanating from the origin. He showed that lines should curve by $\Theta(\log n)$ to satisfy a set of axioms analogous to the Euclid's axioms (the lower bound proof was given by H\r{a}stad). The theory was recently re-discovered by Chun {\em et al.}~\cite{cknt-cdg-09j} and Christ {\em et al.}~\cite{ChristJournal12}. Using these results we can define a geometry that satisfies Euclid-like axioms in the two dimensional grid, and only a small bend of the lines will be needed (i.e., $\Theta(\log n)$ in an $n \times n$ grid, a formal definition is given below).

Chun {\em et al.} and Christ {\em et al.}  proposed a $d$-dimensional version of the set of axioms, but unfortunately it is not constructive. That is, they left open how to find a system to generate a set of digital segments in $d$-dimensional space that resembles the Euclidean segments. In this paper we provide the first significant step towards answering the question for high dimensions. For the purpose we extend the constructive algorithm of Christ {\em et al.}~\cite{ChristJournal12} to spaces of arbitrary dimension and study how much of a bend it creates.

\section{Preliminaries}
Let $x_1, x_2, \ldots, x_d$ denote the coordinate axes in $\Z^d$, and $p_i$ denote the $i$-th coordinate of a point $p\in \Z^d$ (for simplicity, from now on all indices are in the set $\{1, \ldots, d\}$). For any two points $p,q\in\Z^d$, we denote the path connecting $p$ and $q$ by $R(p,q)$. We aim for a constructive method to define a segment from any pair of points $(p,q)\in \Z^d\times \Z^d$. As we will see later, it will be useful to consider a general definition for subsets of $\Z^d \times \Z^d$.

\begin{definition} For any $S\subseteq \Z^d \times \Z^d$, let $DS$ be a set of digital segments such that $(p,q)\in S \rightarrow R(p,q)\in DS$. We say that $DS$ forms a {\em partial set of consistent digital segments} on $S$ (partial CDS for short) if for every pair $(p,q)\in S$ it satisfies the following five axioms:

\begin{itemize}
\item[(S1)] Grid path property: $R(p,q)$ is a path between $p$ and $q$ under the $2d$-neighbor topology\footnote{The $2d$-neighbor topology is the natural one that connects to your predecessor and successor in each dimension. Formally speaking, two points are connected if and only if their $L_1$ distance is exactly one.}. 
\item[(S2)] Symmetry property: $R(p,q)\in DS \rightarrow R(p,q)=R(q,p)$.
\item[(S3)] Subsegment property: For any $r\in R(p,q)$, we have $R(p,r) \in DS$ and $R(p,r) \subseteq R(p,q)$.
\item[(S4)] Prolongation property: There exists $r \in \Z^d$, such that $R(p,r)\in DS$ and $R(p,q) \subset R(p,r)$.
\item[(S5)] Monotonicity property: for all $i\leq d$ such that $p_i = q_i$, it holds that every point $r\in R(p,q)$ satisfies $r_i = p_i = q_i$.
\end{itemize}
\end{definition} 

These axioms give nice properties of digital line segments analogous to Euclidean line segments. For example, (S1) and (S3) imply that the non-empty intersection of two digital line segments is connected under the $2d$-neighbor topology. In particular, the intersection between two digital segments is a digital line segment that could degenerate to a single point or even to an empty set. (S5) implies that the intersection with any axis-aligned halfspace is connected, and so on. 

Our aim is to create a partial CDS on $\Z^d\times \Z^d$. We call such a construction a set of {\em consistent digital segments} (CDS for short). Similarly, a partial CDS on $\{p\}\times \Z^d$ (for some $p\in\Z^d$) is a collection of  segments emanating from $p$ and thus is referred to as a {\em consistent digital ray} system (or CDR for short).

Although conceptually simple, it is not straightforward to create CDSs or CDRs, even when $d=2$. For example, the simple rounding scheme of a Euclidean segment to the digital world that is often used in computer vision, does not generate a CDS (since axioms are not always preserved, see Figure~\ref{fig:rounding}). Another alternative is to use the bounding box approach that makes all moves in one dimension before moving in another one. Although this set of segments is consistent, it will be visually very different from the Euclidean line segments. Thus, the objective is to create a CDR or a CDS that resembles the Euclidean segments. 

The straightness or resemblance between the digital line segment $R(p,q)$ and the Euclidean segment $\overline{pq}$ is often measured with the {\em Hausdorff} distance. The Hausdorff distance $H(A,B)$ of two objects $A$ and $B$ is defined by $ H(A, B) = \max \{ h(A,B), h(B, A) \}$, where $h (A,B)
= \max_{a \in A} \min_{b \in B} \delta(a,b)$, and $\delta(a,b)$ is some fixed underlying distance. In this paper we use the natural Euclidean distance $|| \cdot ||_2$ to measure distance between two points. %

\begin{figure}[h]
\centering
\includegraphics[page=1]{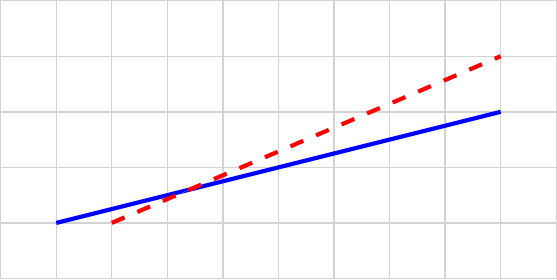}
\includegraphics[page=2]{figure/rounding}
\caption{Two different Euclidean line segments and their corresponding digital line segments via a rounding scheme. Note that their intersection in $\Z^2$ (highlighted with grey disks) is not connected under the 4-neighbor topology, which implies that the rounding scheme is not consistent.}
\label{fig:rounding}
\end{figure}

\begin{definition}
Let $DS(S)$ be a partial CDS. We say that $DS(S)$ has Hausdorff distance $f(n)$ if for all $p, q \in S$ such that $d_2(p,q)\leq n$, it holds that $H(\overline{pq},R(p,q)) = O(f(n))$.
\end{definition} 

Constructions with smaller Hausdorff distance resemble more the Euclidean segments and thus, are more desirable. Hence, the big open problem in the field is what is the (asymptotically speaking) smallest $f(n)$ function so that we can have a CDS in $\Z^d$? Or equivalently: what is the asymptotic behaviour of the Hausdorff distance of the CDS that best approximates the Euclidean segments?

\subsection{Previous work}
Although the concept of consistent digital segments was first studied by Luby~\cite{luby}, it received renewed interest by the community when it was rediscovered by Chun {\em et al.}~\cite{cknt-cdg-09j}. The latter showed how to construct a set of consistent digital rays (CDR) in any dimension. The construction and satisfies all axioms, including the Hausdorff distance bound:

\begin{theorem}[Theorem~4.4 of \cite{cknt-cdg-09j}, rephrased]
\label{theo:chunCDR}
For any $d\geq 2$ and $p\in\Z^d$ we can construct a CDR with $O(\log n)$ Hausdorff distance. 
\end{theorem}

H\r{a}stad~\cite{luby} and Chun {\em et al.}~\cite{cknt-cdg-09j} showed that any CDR in two dimensions must have $\Omega(\log n)$ Hausdorff distance. Thus the $\log n$ is the smallest possible distance one can hope to achieve. 
This result was generalized by Christ {\em et al.}~\cite{ChristJournal12}, where they show a correspondence between CDRs in $\Z^2$ and total orders on the integers (details on this correspondence is given in Section~\ref{sec_construct}). In particular, this correspondence can be used to create a CDS in $\Z^2$ that has $O(\log n)$ Hausdorff distance. Note that the $\Omega(\log n)$ lower bound also holds for CDS, so this result is asymptotically tight. 

This answers the question of how well can CDSs approximate Euclidean segments in the two dimensional case. However, the question for higher dimensions remains largely open. Although the method of Christ {\em et al.}~\cite{ChristJournal12} cannot be used to construct CDSs or CDRs in high dimensions, they show that it can create partial CDS as follows.

\begin{theorem}[Theorem~16 of \cite{ChristJournal12}, rephrased]
\label{thm:christ_2}
Let $S=\{(x,y) \colon x_i\geq y_i\} \subset \Z^d\times\Z^d$. Then, we can construct arbitrarily many partial CDS on $S$.
\end{theorem}

Note that $S$ contains segments of positive slope (that is, only for the pairs $(p,q)$ such that $q$ is in the first orthant of $p$), hence it is roughly a small fraction (roughly $1/2^{d-1}$) of all possible segments. Other than Theorems~\ref{theo:chunCDR} and~\ref{thm:christ_2}, little or nothing is known for three or higher dimensions. Up to the date, the only CDS known in three or higher dimensions is the naive bounding box approach (described in Section~\ref{sec_construct}) that has $\Omega(n)$ Hausdorff distance. In particular, it still remains open whether one can create a CDS in $\Z^d$ with $o(n)$ Hausdorff distance (for $d>2$). 

Other research in the topic has focused in the characterization of CDSs in two dimensions. Chowdury and Gibson~\cite{Chowdhury2015} gave necessary and sufficient conditions so that the union of CDRs forms a CDS. This characterization heavily uses the correspondence between CDRs and total orders, and thus it was stated in terms of total orders. %
In a companion paper, the same authors~\cite{Chowdhury2016} afterwards provided an alternative characterization together with a constructive algorithm. Specifically, they gave an algorithm that, given a collection of segments in an $n \times n$ grid that satisfies the five axioms, computes a CDS that contains those segments. The algorithm runs in polynomial time of $n$. Both results only hold for the two dimensional case.

\paragraph*{Other definitions}
Given two points $p,q\in\Z^d$ such that $p\neq q$, the {\em slope} of $R(p,q)$ is the sign vector $\mat{t}=(t_1, t_2, \ldots, t_d) \in \{+1,-1\}^d$, where $t_i = +1$ if $p_i \leq q_i$ and is $-1$ if $p_i \geq q_i$. For simplicity, along the paper we talk about {\em the} slope of $R(p,q)$ (whenever $p$ and $q$ have more than one slope we pick one arbitrarily). Let $T$ be the set containing all $2^d$ slopes of $\Z^d$.  %

A {\em total order} $\theta$ of $\Z$ is a binary relation on all pairs of integers. We to denote that $a$ is smaller than $b$ with respect to $\theta$ by $a \prec_\theta b$. We say that two elements $a$ and $b$ are {\em consecutive} if there is no number between them (i.e., no integer $c$ satisfies $a \prec_\theta c \prec_\theta b$).

We define three operations on total orders: shift, flip and reverse. The {\em shift} operation is denoted $\theta + c$ and is the result of adding a constant value $c$ to each integer without changing their binary relations (that is, $a \prec_\theta b$ if and only if $a + c \prec_{\theta + c} b + c$). Similarly, {\em flipping} is denoted by $- \theta$ and is the result of changing the sign of all binary relations (that is, $a \prec_\theta b$ if and only if $-a \prec_{-\theta} -b$). The {\em reverse} operation of $\theta$ (denoted by $\theta^{-1}$) is the total order resulting in inverting all relationships (that is, $a\prec_{\theta} b$ 
if and only if
$b \prec_{\theta^{-1}} a$). 

Sometimes we will restrict a total order $\theta$ to an interval $[a,b]$. We denote this by $\theta [a,b]$. For these subsets we also use the same shift, flip and reverse operations whose definitions follow naturally. As an example, observe that $\theta[a,b] + c = (\theta+c)[a+c, b+c]$. Along the paper we will associate a total order to a point $p$ and a slope $\mat{t}$. This will be denoted by $\theta^p_\mat{t}$. We will omit the subscript or superscript if it is clear from the context or we use the same total order for all slopes or points, accordingly. Due to lack of space some proofs are deferred to the Appendix. Whenever possible, we added a sketch of the proof in the main document.

\subsection{Paper organization}\label{sec_organiz}
We study properties that CDRs and CDSs must satisfy in high dimensions (i.e., $d\geq 3$), and show that they behave very differently from the two-dimensional counterparts. In Section~\ref{sec_construct} we introduce the concept of {\em axis-order}. Although not needed in two dimensions, it allows us to extend the total order construction of Christ {\em et al.} to higher dimensions. Given a point $p\in \Z^d$, a total order $\theta$ on the integers, a slope $\mat{t}$, we construct a partial CDS which we denote by $\CDO(\theta,p,\mat{t})$. Specifically, it contains segments having an endpoint $p$ and slope $\mat{t}$ (that is, an orthant whose apex is $p$). In order to create a CDR, we combine $2^d$ such constructions (one for each slope), and characterize when will such approach work.

\begin{theorem}[Necessary and sufficient condition for CDRs]\label{theo:CDR}
For any $d>2$, point $p \in \Z^d$ and set $\{ \theta_{\mat{t}} \colon t\in T\}$ of $2^d$ total orders, $\bigcup_{\mat{t}\in T}\CDO(\theta_{\mat{t}},p,\mat{t})$ forms a CDR at $p$ if and only if for any $\mat{t},\mat{t}'\in T$ it holds that $\theta_\mat{t}[\mat{t}\cdot p, \infty) = \theta_{\mat{t}'}[ \mat{t}' \cdot p, \infty) - \mat{t}' \cdot p + \mat{t}\cdot p$.
\end{theorem}

This result highly contrasts with the two dimensional counterpart of Christ {\em et al.}~\cite{ChristJournal12}: in two dimensions we have four different slopes (and thus, four associated quadrants). We can use four different total orders (one for each of the quadrants) and the union will always be a CDR. In higher dimensions this is not true: fixing the total order for a single orthant uniquely determines the behaviour of other orthants. %
In particular, there is a unique way of completing the partial CDS $\CDO(\theta,p,\mat{t})$ to a CDR which we denote by $\CDS(\theta,p)$. 

The next step is to consider the union of several CDRs to obtain a CDS. In Section~\ref{sec:nec} we characterize for which total orders this is possible. 

\begin{theorem}[Necessary and sufficient condition for CDSs]\label{theo_characCDS}
$\theta$ is a total order such that $\bigcup_{p\in\Z^d}\CDR(\theta,p)$ forms a CDS if and only if $\theta = \theta + 2$ and $\theta = -(\theta + 1)^{-1}$.
\end{theorem}

This result also contrasts with the two dimensional case: if we replicate the same construction for all points of $\Z^2$, the result will always be a CDS for any total order. However, in higher dimensions this will only hold for some total orders. In particular, this result positively answers the question posed by Christ {\em et al.} of whether their approach can be extended to create CDSs in higher dimensions~\cite{christarXiv}.

The main difference between two and higher dimensional spaces is that the construction for two different slopes have a larger portion in common. In two dimensions, two quadrants share at most a line (whose behaviour is unique because of the monotonicity axiom), but in general orthants may share a subspace of dimension $d-1$. The total orders associated to each orthant must behave similarly within the subspace, which creates some dependency between the total orders. More importantly, each orthant shares subspaces with other orthants, and so on. This cascades creating common dependencies that cycle back to the original orthant and highly constrain the total orders. In Section~\ref{sec:concl} we discuss this dependency and argue that variations of this construction will also have the same necessary and sufficient conditions.

\section{Extending the total order construction to high dimensions}
\label{sec_construct}
In this section we use a total order to construct a CDR in $\Z^d$. We start by reviewing the construction of Christ {\em et al.}~\cite{ChristJournal12} for $\Z^2$. Given a total order $\theta$ and two points $p=(p_1,p_2),q=(q_1,q_2) \in \Z^2$ such that $q_1\geq p_1$ and $q_2 \geq p_2$, we view the digital segment $R(p,q)$ as a collection of steps that form a path from $p$ to $q$. Due to the monotonicity property, in each step the path increases either the first or second coordinate by one. Clearly, this path must do $q_1+q_2-p_1-p_2$ steps, out of which $q_1-p_1$ are in the $x_1$ coordinate (and the remaining ones in the $x_2$ coordinate). The choice of which steps we move in which coordinate depends on $\theta$: assume that after moving several steps we have reached some intermediate point $(r_1,r_2)$. Then, we check whether or not the number $r_1+r_2$ is among the $q_1-p_1$ smallest elements of $\theta[p_1+p_2,q_1+q_2-1 ]$. If so, we move from $(r_1,r_2)$ in the $x_1$ coordinate. Otherwise we do so in the $x_2$ coordinate (see an example in Figure~\ref{fig:construction}).

\begin{figure}[h]
\centering
\includegraphics{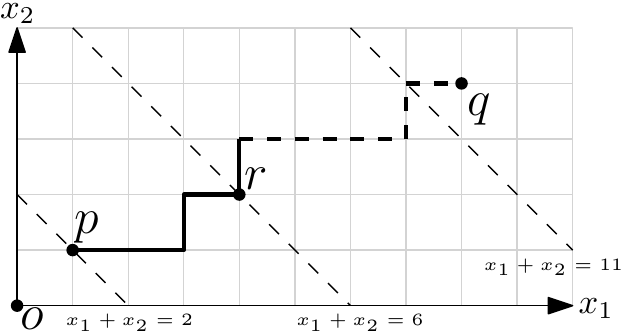}
\caption{Example of the construction of Christ {\em et al.} in $\Z^2$. Given $p=(1,1)$, $q=(8,4)$ and a total order $\theta$ such that $\theta[2,11] = 5 \prec 3 \prec 2 \prec 7 \prec 9 \prec 8 \prec 11 \prec 10 \prec 6 \prec 4$. The path must perform $q_{1}- p_{1}=7$ steps in the $x_1$ direction and $q_{2}- p_{2}=3$ steps in the $x_2$ direction. Since $p_{1}+p_{2} = 2$ and $2$ is among the $7$ smallest elements in $\theta[2,11]$, it moves in the $x_1$ direction. Similarly, at point $r=(4,2)$, the path will move in $x_2$ direction because $r_{1}+r_{2} =6$ is among the $3$ largest elements of $\theta[2,11]$. 
Observe that, for any $c \in [2,11]$ there is a unique point $m$ in the path such that $m_1+m_2=c$.}
\label{fig:construction}
\end{figure}

All of the segments created this way have slope $(+1,+1)$. In a similar way, we can pick a total order to define the segments emanating from $p$ with slope $(+1,-1)$, $(-1,+1)$ and $(-1,-1)$. We emphasize that there is no dependency between the total orders: the choice of total order for one slope has no impact on the available options for the others. Moreover, any four choices will result in a CDR (Similarly, any CDR in $\Z^d$ is associated to $2^d$ total orders of $\Z$, one for each slope). As mentioned before, this independence between quadrants does not hold in higher dimensions.

\subsection{Constructing a CDR in $\Z^d$ from a total order}
The construction of Christ {\em et al.} explains how to construct segments of slope $(+1,+1)$ in $\mathbb{Z}^2$ (or equivalently, for points in the first quadrant). The segments of different slopes are obtained via symmetry. In higher dimensions it will be useful to have an explicit way to construct segments of any slope. Thus, we first generalize the method of Christ {\em et al.} for any orthant. 

In order to get an idea of our approach, we first look at the folklore bounding box approach to construct a CDS. When traveling from point $p$ to point $q$, we consider the bounding box formed by the two points. The point with smaller $x_1$ coordinate will move in the $x_1$ coordinate until reaching the $x_1$ coordinate of another point. Afterwards, the one with smaller $x_2$ coordinate will move in the $x_2$ coordinate, and so on until the two points meet (see Figure~\ref{fig:bounding}).

\begin{figure}[h]
\centering
\scalebox{0.5}{\includegraphics{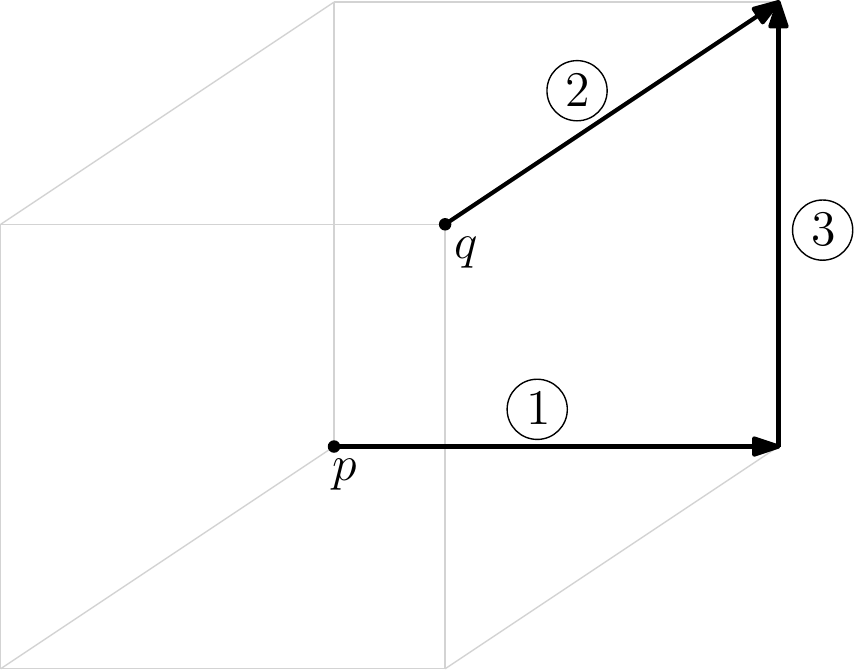}}
\caption{Example of the bounding box approach in $\Z^3$. $p=(0,3,0)$ and $q=(3,0,3)$. The number in each circle indicates the order in which we execute the movements.}
\label{fig:bounding}
\end{figure}

So, if $d=3$, for any segment whose slope is $(+1,+1,+1)$ we first do all the movements in the $x_1$ coordinate, then $x_2$ coordinate, and finally in the $x_3$ coordinate. However, if the segment has slope $(+1,-1,-1)$, then the bounding box CDS will travel first in the $x_1$ coordinate, then $x_3$ and finally $x_2$. Intuitively speaking, even though in both cases we are performing the same steps (i.e, we use the natural order  $0 \prec 1 \prec 2 \prec 3 \prec \ldots $), the order in which we execute each dimension is slightly different (or equivalently, the total order is being interpreted differently). 
We model this difference in interpretation through a new concept which we call {\em axis-order}.

Given a slope $(t_1, t_2, \ldots, t_d)$, let $a_1, \ldots a_k$ be indices of the coordinates with positive value in increasing order (that is, $t_i=+1$ if and only if $i=a_j$ for some $j\leq k$). Similarly, let $b_1, \ldots b_{d-k}$ be the indices of the the coordinates with negative value in decreasing order. Then, the {\em axis-order} of $(t_1, t_2, \ldots, t_d)$ is $x_{a_1}, x_{a_2}, \ldots, x_{a_k}, x_{b_1}, \ldots, x_{b_{d-k}}$. For example, the axis-order of $(-1,+1,+1)$ is $x_2,x_3,x_1$, and the axis-order of $(+1,-1,+1)$ is $x_1,x_3,x_2$. As we will see later, it will be useful to consider subspaces of $\Z^d$. We observe a property that follows from the definition of axis-order.

\begin{observation}
\label{obs:same_order}
Let $a_1, \ldots a_k$ be a sequence of indices such that $a_1<\ldots<a_k$, and let $\mat{t},\mat{t}' \in \{-1,1\}^d$ be two slopes such that $\mat{t}_{a_i}=\mat{t}'_{a_i}$ (for all $i\leq k$). Then, $\mat{t}$ and $\mat{t}'$ have the same axis-order $\tau$ restricted to a subspace $\mathcal{H}$ spanned by $\{x_{a_{1}}, x_{a_{2}}, \ldots, x_{a_k}\}$. Moreover, the axis-order of $-\mat{t}$ and $-\mat{t}'$ restricted to $\mathcal{H}$ is the reverse of $\tau$.
\end{observation}

With the help of axis-order we can extend the two dimensional construction to higher dimensions. Given a point $p=(p_1,\ldots, p_d)\in \Z^d$, a total order $\theta$ and a slope $\mat{t}$, we construct the set of rays emanating from $p$ with that slope. Let $\mathcal{O}_{\mat{t}}(p)=\{q \in \Z^d \colon \mat{t}_i\cdot q_i \geq \mat{t}_i \cdot p_i\}$: by definition, the segment from $p$ to any point in $\mathcal{O}_{\mat{t}}(p)$ has slope $\mat{t}$. Also, let  $x_{a_1}, x_{a_2}, \ldots$ be the axis-order of $\mat{t}$. 

For any point $q=(q_1, \ldots, q_d) \in \mathcal{O}_{\mat{t}}(p)$ we construct the segment $R(p,q)$. Similar to the two dimensional case, the path from $p$ to $q$ must do $\mat{t} \cdot q-\mat{t} \cdot p$ steps, out of which $|p_1-q_1|$ will be in the first coordinate, $|p_2-q_2|$ in the second, and so on. We traverse through intermediate points, each time increasing the inner product with $\mat{t}$ by one. At each intermediate point $r$, we check the position of $\mat{t} \cdot r$ in $\theta[\mat{t} \cdot p, \mat{t} \cdot q-1]$; if it is among the $|p_{a_1}-q_{a_1}|$ smallest elements in $\theta[\mat{t} \cdot p, \mat{t} \cdot q-1]$ then we move in the $x_{a_1}$ coordinate. Otherwise, if it is among the smallest $|p_{a_1}-q_{a_1}|+|p_{a_2}-q_{a_2}|$ elements we move in $x_{a_2}$, and so on. 

For example, if the total order $\theta$ satisfies $3\prec_\theta 1 \prec_\theta 5 \prec_\theta 7 \prec_\theta 9 \prec_\theta 8 \prec_\theta 6 \prec_\theta 4 \prec_\theta 2 \prec_\theta 0$, $p=(0,0,0)$ and $q=(2,-3,5)$, the slope is $(+1,-1,+1)$, axis-order is $x_1, x_3, x_2$. So we must look at $\theta[p\cdot(+1,-1,+1),q\cdot(+1,-1,+1)-1]=\theta[0,9]$. In this total order the number $(+1,-1,+1)\cdot (0,0,0)=0$ is the largest element in $\theta[0,9]$, so we move from $(0,0,0)$ in the $x_2$ coordinate to point $(0,-1,0)$. At point $(0,-1,0)$ the number $(+1,-1,+1)\cdot (0,-1,0)=1$ is the second smallest element in $\theta[0,9]$, so we move in the $x_1$ coordinate, and so on. Overall the path is $(0,0,0) \rightarrow (0,-1,0) \rightarrow (1,-1,0) \rightarrow (1,-2,0) \rightarrow (2,-2,0) \rightarrow (2,-3,0) \rightarrow (2,-3,1) \rightarrow (2,-3,2) \rightarrow (2,-3,3) \rightarrow (2,-3,4) \rightarrow (2,-3,5)$.

\begin{definition}
For any point $p\in \Z^d$, slope $\mat{t}$, and total order $\theta$, we call the collection of segments $\{R(p,q) \colon q\in \mathcal{O}_{\mat{t}}(p)\}$ the {\em total order} construction of $\theta$ (centered at $p$) for the slope $\mat{t}$, and denote it by $\CDO(\theta,p,\mat{t})$.
\end{definition}

\subsection{Properties of the total order construction}

\begin{lemma}[Translation Lemma]
\label{lem:DO_tran}
For any $p \in \Z^d$, slope $\mat{t}$ and total order $\theta$, the set of segments in $\CDO(\theta, p, \mat{t})$ is the translated copy of the set of segments in $\CDO(\theta - \mat{t}\cdot p, o, \mat{t})$, where $o$ is the origin.
\end{lemma}

\begin{lemma}
\label{lem:CDO}
For any $p\in \Z^d$, slope $\mat{t}$ and total order $\theta$, the set of segments in $\CDO(\theta,p,\mat{t})$ forms a partial CDS on $\{p\}\times \mathcal{O}_{\mat{t}}(p)$. 
\end{lemma}
\begin{proof}
This statement is a particular case of of Theorem~\ref{thm:christ_2}: we are interested in segments of a single slope emanating from a fixed point, whereas Theorem~\ref{thm:christ_2} only requires segments of a fixed slope). The proof given by Christ {\em et al.}~\cite{ChristJournal12} is for slope $(+1, \ldots , +1)$, but the arguments extend naturally for the general case.
\end{proof}

Let $\theta_0$ be the natural order on the integers (that is, $\theta_0 = \{\ldots\prec  -1 \prec 0\prec 1\prec 2\prec \ldots\}$). Fix any point $p\in \Z^d$ and apply the total order construction $\CDO(\theta,p,\mat{t})$ to all slopes. Let $\CDR(\theta_0,p)$ be the union of all segments created. Similarly, let $\theta_1$ be result of swapping the position of $-1$ and $-2$ in $\theta_0$ (i.e., $\theta_1=\{\ldots \prec -1 \prec -2 \prec 0\prec 1\prec 2 \ldots\}$). And let $\CDR(\theta_1,p)$ be the union of all segments created when using $\theta_1$ instead.

\begin{proposition}\label{prop_nocdr}
$\CDR(\theta_0,p)$ is a CDR that is included in the bounding box CDS whereas $\CDR(\theta_1,p)$ is not a CDR. 
\end{proposition}

\subsection{Gluing orthants to obtain CDRs}
The second example of Proposition~\ref{prop_nocdr} shows an example of a  total order that cannot be applied everywhere to form a CDR. %
Theorem~\ref{theo:CDR} stated in Section~\ref{sec_organiz} shows the relationship that total orders in different slopes must satisfy in order to create a CDR. Intuitively speaking, this correlation is so big that choosing one total order effectively fixes the rest. The remainder of this section is dedicated to proving this interdependency. We start by showing one side of the implication.

\begin{lemma}[Necessary condition for CDRs]\label{lem:necCDR}
Let $p \in \Z^d$ and $\{ \theta_{\mat{t}} \colon t\in T\}$ be a set of $2^d$ total orders such that $\bigcup_{\mat{t}\in T}\CDO(\theta_{\mat{t}},p,\mat{t})$ forms a CDR. Then, for any $\mat{t},\mat{t}'\in T$, it holds that $\theta_\mat{t}[\mat{t}\cdot p, \infty) = \theta_{\mat{t}'}[ \mat{t}' \cdot p, \infty) - \mat{t}' \cdot p + \mat{t}\cdot p$.
\end{lemma}
\begin{proof}[Proof (sketch)]
We prove the statement by contradiction. That is, assume that there exist two slopes $\mat{t}, \mat{t}'$ such that $v \prec_{\theta_\mat{t}} v' $ but $v' - \mat{t} \cdot p + \mat{t}'\cdot p \prec_{\theta_{\mat{t}'}} v - \mat{t} \cdot p + \mat{t}'\cdot p $. Without loss of generality, we can choose $\mat{t}$ and $\mat{t}'$ so that they share a plane (pick a sequence of intermediate orthants so that pairwise they do, and look at the first time in which the equality is not satisfied).
We pick a point $q$ such that $R(p,q)$ has both slope $\mat{t}$ and $\mat{t}'$, and look at $R(p,q)$ from both the viewpoints of $\CDO(\theta_\mat{t}, p, \mat{t})$ and $\CDO(\theta_{\mat{t}'}, p, \mat{t}')$. 

Along the path $R(p,q)$ we look at two intermediate points $r$ and $r'$. The main feature of these points is that the behaviour of $R(p,q)$ at those points depends on the positions of $v$ and $v'$ in $\theta_\mat{t}$ (if we look at it from the viewpoint of $\CDO(\theta_\mat{t}, p, \mat{t})$). Since $v \prec_{\theta_\mat{t}} v' $, we can choose $q$ in a way that the path will move in different directions at the two points. Then, we study the same segment from the viewpoint of the other orthant. In this case, the behaviour of the same intermediate points will depend on the positions of $v' - \mat{t} \cdot p + \mat{t}'\cdot p$ and $v - \mat{t} \cdot p + \mat{t}'\cdot p $ in the shifted total order instead. Thus, if the relationships are reversed, the two paths behave differently and in particular we cannot have a CDR.
\end{proof}

\begin{lemma}[Sufficient condition for CDRs]\label{lem_sufficient}
For any point $p \in \Z^d$, let $\{ \theta_{\mat{t}} \colon t\in T\}$ be a set of $2^d$ total orders such that $\theta_\mat{t}[\mat{t}\cdot p, \infty) = \theta_{\mat{t}'}[ \mat{t}' \cdot p, \infty) - \mat{t}' \cdot p + \mat{t} \cdot p$ for any $\mat{t},\mat{t}'\in T$. Then, $\bigcup_{\mat{t}\in T}\CDO(\theta_{\mat{t}},p,\mat{t})$ forms a CDR. 
\end{lemma}

This completely characterizes the CDRs that can be made with the total order construction in $\Z^d$. %
For any point $p$, slope $\mat{t}$ and total order $\theta$, there is a unique CDR that can be created in this way and contains $\CDO(\theta,p,\mat{t})$. 
Since the choice of slope is not important, let $\CDR(\theta,p)$ be the unique CDR that contains $\CDO(\theta,p,(+1,\ldots, +1))$.

\begin{corollary}
For any $p\in\Z^d$ there exist arbitrarily many CDRs with $O(\log n)$ Hausdorff distance. 
\end{corollary}
\begin{proof}
An explicit construction of a single CDR in $\Z^d$ with $O(\log n)$ Hausdorff distance was given by Chun {\em et al.}~\cite{cknt-cdg-09j}. They showed that the CDR generated using the Van der Corput sequence~\cite{c-v-35} as total order has low Hausdorff distance (for any dimension). Christ {\em et al.}~\cite{ChristJournal12} extended the result showing that the straightness is asymptotically same as the discrepancy of the permutation corresponding to the total order, which is known to be $\Theta (\log n)$.
 The arguments for $d=2$ to our higher dimension construction can be directly applied. Thus, we omit them.
\end{proof}

\section{Necessary and sufficient conditions for CDSs}
\label{sec:nec}

Next we focus our attention to constructing CDSs. Christ {\em et al.}~\cite{ChristJournal12} showed that if we apply the same total order construction to all points of $\Z^2$ we get a collection of CDRs whose union is always a CDS. For any total order $\theta$, let $\CDS(\theta)=\bigcup_{p\in\Z^d}\CDR(\theta,p)$. Unlike the two dimensional case, the construction $\CDS(\theta)$ does not always yield a CDS in higher dimensions. Theorem~\ref{theo_characCDS} stated in Section~\ref{sec_organiz} gives necessary and sufficient conditions that the total order must satisfy.

Recall that in principle, we allow different orthants (except $(+1, \ldots, +1)$) to have different total orders in this construction. For any point $p\in \Z^d$ and slope $\mat{t}$, let $\theta^p_{\mat{t}}$ be the total order associated to point $p$ and slope $\mat{t}$ in $\CDS(\theta)$. Since $\CDS(\theta)$ in particular contains $\CDS(\theta,p)$, Theorem~\ref{theo:CDR} gives a relationship between $\theta$ and $\theta^p_{\mat{t}}$. We give a stronger bound on that relationship as well.

\begin{theorem}\label{theo_shiftedCDS}
If $\theta$ is a total order such that $\CDS(\theta)$ forms a CDS, then for any $p\in \Z^d$ and slope $\mat{t}$ it holds that $\theta^p_{\mat{t}}[\mat{t}\cdot p, \infty)=\theta[\mat{t}\cdot p, \infty)$. In particular, $\CDS(\theta^p_{\mat{t}},p,\mat{t})=\CDS(\theta,p,\mat{t})$.
\end{theorem}

This shows that, if we want to create a CDS in this fashion, we must use the same total order $\theta$ for all points and all slopes. Again, this contrasts with the $d=2$ case where we can combine any two total orders for slopes $(+1,+1)$ and $(+1, -1)$. Christ {\em et al.}~\cite{ChristJournal12} showed that if we repeat the construction for all points of $\Z^2$ the union will form a CDS. The remainder of this section is dedicated to showing Theorems~\ref{theo_characCDS} and~\ref{theo_shiftedCDS}.

\subsection{Two dimensional preliminaries}
Along the proof, we will often consider two dimensional subspaces and find some requirements that extend to the whole space. Thus, we first show a subtree property that CDS in $\Z^2$ must satisfy. Consider any point $p\in\Z^2$, slope $\mat{t}$, point $q\in \CDO(\theta^p_{\mat{t}},p,\mat{t})$ such that $q\neq p$, and all points $r\in\Z^2$ such that $R(p,r)$  passes through $q$. This set of points (and their paths to $q$) form a subtree of $\CDO(\theta^p_{\mat{t}}, {p},\mat{t})$. The same tree must be part of $\CDO(\theta^q_{\mat{t}}, {q},\mat{t})$ or it would will violate (S3) (see Figure~\ref{fig:subtree}, left).

We express this subtree property in terms of total orders $\theta^p_\mat{t}$ and $\theta^q_\mat{t}$. Assume $\mat{t}=(+1,+1)$, let $s_1,s_2 \geq 0$ be integers such that $q=p+(s_1,s_2)$, and let $n$ be any number such that $n>s_1+s_2$. We will consider the restriction of the total order $\theta^p_\mat{t}$ to three intervals: $[\mat{t} \cdot p, \mat{t} \cdot q  -1]$, $[\mat{t} \cdot q, \mat{t} \cdot p+n  -1]$, and 
$[\mat{t} \cdot p, \mat{t} \cdot p +n -1]$. Note that the union of the first two forms the third one. In order to reduce notation we call them the left, right, and complete intervals. Similarly, we call $\theta^p_\mat{t}[\mat{t} \cdot p, \mat{t} \cdot q  -1]$, $\theta^p_\mat{t}[\mat{t} \cdot q, \mat{t} \cdot p+n  -1]$, and 
$\theta^p_\mat{t}[\mat{t} \cdot p, \mat{t} \cdot p +n -1]$ the {\em left order}, the {\em right order} and the {\em complete order}. The subtree property says that many inequalities in the right order must also be held in $\theta^q_\mat{t}$.

First assume that $s_1, s_2 \neq 0$; let $a$ and $b$ be the $s_1$-th and $(s_1+1)$-th smallest numbers in the left order, respectively. By definition, these two numbers are consecutive in the left order, but they need not be in the complete order (i.e., there could be numbers from the right interval). 

Let $i_a$ and $i_b$ be the positions of $a$ and $b$ in the complete order, respectively. We partition the numbers of the large interval into three groups, depending on whether they are $(i)$ smaller than $a$, $(ii)$ larger than $a$ and smaller than $b$, or $(iii)$ larger than $b$ (all these comparisons are with respect to $\theta^p_\mat{t}$). Let $X_1(n)$, $X_2(n)$, and $X_3(n)$ be the three sets, respectively (see Figure~\ref{fig:subtree}).

\begin{figure}
\centering
\scalebox{0.7}{\includegraphics{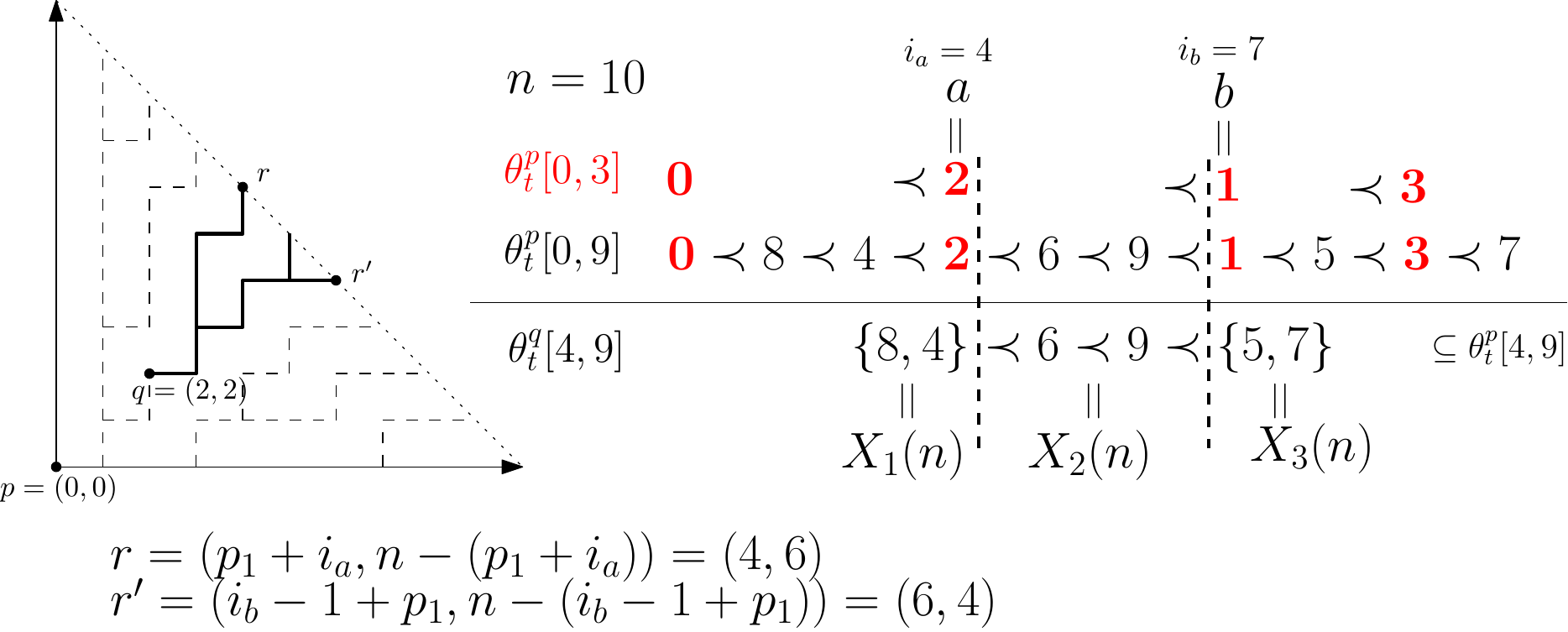}}
\caption{Example of the subtree property.
(left) geometric interpretation of the subtree property. The paths to $p$ that pass through $q$ impose a constraint on $\theta^q_{\mat{t}}$. In particular, a point in the diagonal $x_1+x_2=n$ will pass through $q$ if and only if it is between $r$ and $r'$ (highlighted points in the figure).
(right) implications in the total order of $\theta^q_{\mat{t}}$. In red bold we highlight the points that belong to the left interval. The points in the right interval are classified into the three sets $X_1(n)$, $X_2(n)$ and $X_3(n)$ according to their positions (left of $a$, right of $b$, or in between). The fact that the subtree of $q$ (black in the left figure) has to be preserved in $q$ implies many relationships for $\theta^q_{\mat{t}}$ that are shown in the third line.
}
\label{fig:subtree}
\end{figure}

Before giving the subtree property we extend the definitions of these three sets for the cases in which $s_1$ and $s_2$ can be zero. If $s_1=0$ then $a$ and $i_a$ are not well defined (similarly, $b$ and $i_b$ are not defined when $s_2=0$). In the first case we set $i_a=0$, $X_1(n)=\emptyset$ and classify the numbers of the right interval into $X_2(n)$ and $X_3(n)$ depending on whether they are smaller or larger than $b$. Similarly, if $i_b$ is not defined, we set $i_b=n+1$, $X_3(n) = \emptyset$, and numbers are be split into the two sets $X_1(n)$ and $X_2(n)$.

The following lemma characterizes the points whose path to $p$ passes through $q$ in the quadrant of $(+1,+1)$.

\begin{lemma}
\label{lem:prefix}
For any $n > s_1 + s_2$, let $r\in \Z^2$ be a point such that $r_1 + r_2 = p_1 + p_2 + n$. The path $R(p,r)$ passes through $q$ if and only if $r_1 \geq q_1$, $r_2 \geq q_2$ and $i_a \leq r_1 - p_1 \leq i_b -1$.
\end{lemma}

\begin{lemma}[The subtree property]\label{lem_subtree}
For any $n>s_1+s_2$ and $u,v\in [\mat{t} \cdot q, \mat{t} \cdot p+n  -1]$, the following relationships must hold in $\theta^q_\mat{t}$.
\begin{itemize}
\item $u \prec_{\theta^q_\mat{t}} v$ for all $u\in X_1(n)$ and $v\in X_2(n)$
\item $u \prec_{\theta^q_\mat{t}} v$ for all $u\in X_1(n) \cup X_2(n)$ and $v\in X_3(n)$
\item $u \prec_{\theta^q_\mat{t}} v$ for all $u,v\in X_2(n)$ such that $u \prec_{\theta^p_\mat{t}} v$
\end{itemize}
\end{lemma}

{\bf Remark} Although we have stated the subtree property for slope $(+1,+1)$, it is straightforward to see that this result extends to other ones. We stick to this notation for simplicity in the exposition, although we will afterwards use it for negative slope as well. 

\subsection{Application in high dimensional spaces}
With the subtree property we can show the first necessary condition of Theorem~\ref{theo_characCDS}.
\begin{lemma}
\label{lem:nec_1}
Let $\theta$ be a total order such that $\CDS(\theta)$ forms a CDS. Then, 
$\theta = \theta + 2$.
\end{lemma}
\begin{proof}
We first give a birdseye overview of the construction: choose an arbitrary $\lambda\in\Z$ and consider the affine plane $\mathcal{H} = \{x_3 = \lambda, x_4=0, \ldots, x_d=0\}$. In this plane we look at the origin $p=(0,0)$, points $q=(-1,0)$ and $r=(0,-1)$ (see Figure~\ref{fig:same_theta}, left). In particular, we look at the third quadrant (the one with slope $(-1,-1)$): first, from Theorem~\ref{theo:CDR} we know that $\theta^p_{(-1,-1)}$ must coincide with $\theta$ (on the interval $[\lambda, \infty)$). 

We apply the subtree property from $p$ to $q$ and $r$; the key property is that both $\theta^q_{(-1,-1)}$ and $\theta^r_{(-1,-1)}$ coincide with $\theta+2$ on the interval $[\lambda+1, \infty)$. Moreover, all paths to $p$ must pass through either $q$ or $r$, which in particular implies that all inequalities from $\theta^p_{(-1,-1)}$ must also be preserved in either  $\theta^q_{(-1,-1)}$ or $\theta^r_{(-1,-1)}$. By combining all of these properties, we show that $\theta$ coincides with $\theta+2$ on the interval $[\lambda+1, \infty)$. The result works for any value of $\lambda$, so when $\lambda \rightarrow -\infty$ we get $\theta=\theta+2$ as claimed.

More formally, pick any $\lambda\in\Z$ and consider the points $p=(0, 0, \lambda, 0, \ldots 0)$, $q=(0, -1, \lambda, 0, \ldots,0)$ and $r=(-1, 0, \lambda, 0, 0, \ldots,0)$. By construction, these points lie on the affine plane $\mathcal{H} = \{x_3 = \lambda, x_4=0, \ldots, x_d=0\}$ as claimed. %

\begin{figure}[h]
\centering
\scalebox{0.5}{\includegraphics{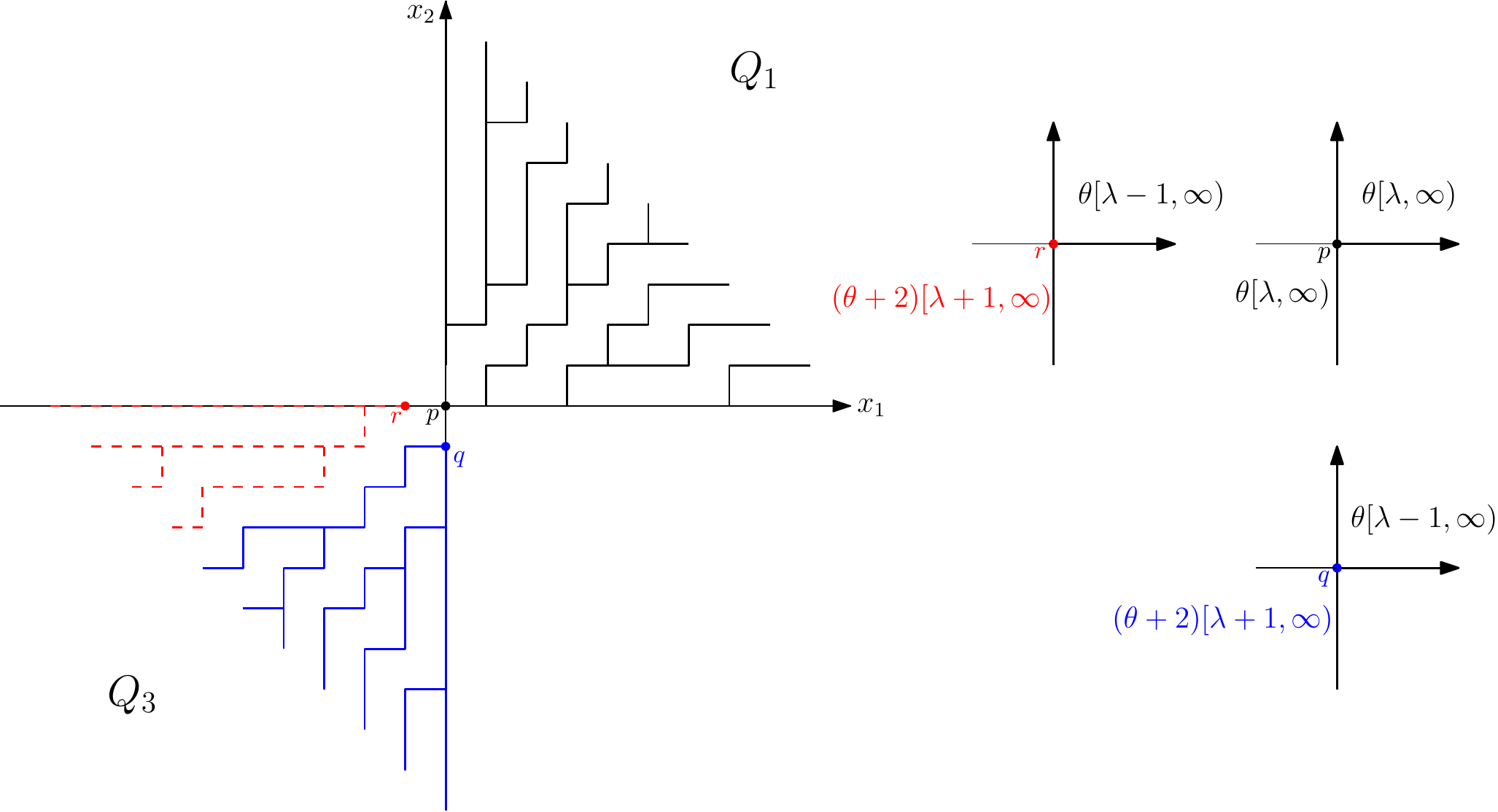}}
\caption{An example of the CDR at $p$ is shown on the left hand side and the total orders applied at $p$, $q$ and $r$ are shown on the right hand side.  
The subtrees at $q$ and at $r$ in $Q_3$ are represented by solid blue and dashed red segments respectively.
}
\label{fig:same_theta}
\end{figure}

Let $\mat{t} = (+1, \ldots, +1)$ and $\mat{t}' = (-1, -1, +1, \ldots, +1)$. By definition of $\CDO(\theta)$ we have $\theta^p_{\mat{t}}=\theta^q_{\mat{t}}=\theta^r_{\mat{t}}=\theta$. We use Theorem~\ref{theo:CDR} to determine the total order used at slope $\mat{t}'$ for the three points: $\theta^p_{\mat{t}'}[\mat{t}' \cdot p, \infty) = \theta^p_{\mat{t}}[\mat{t} \cdot p, \infty) - \mat{t} \cdot p + \mat{t}' \cdot p = \theta[\mat{t} \cdot p, \infty) - \mat{t} \cdot p + \mat{t}' \cdot p = \theta[\lambda, \infty)$. Similarly, at point $q$ we have $\theta^q_{\mat{t}'}[\lambda+1, \infty) = \theta[\lambda-1, \infty) + 2 = (\theta+2)[\lambda+1, \infty)$ and at point $r$ we have $\theta^r_{\mat{t}'}[\lambda+1, \infty) = (\theta+2)[\lambda+1, \infty)$ (The 6 total orders and their relevant orthants are depicted in Figure~\ref{fig:same_theta}, right). %

For any $n>0$ consider the bounded interval $[\lambda, \lambda+n-1]$. 
We apply Lemma~\ref{lem_subtree} in the third quadrant to obtain relationships between $\theta^p_{\mat{t}'}$, $\theta^q_{\mat{t}'}$ and $\theta^r_{\mat{t}'}$. Let $X^{pq}_1(n),X^{pq}_2(n)$, and $X^{pq}_3(n)$ be the partition in the three sets obtained when applying the subtree property to $p$ and $q$ (similarly, we define the sets $X^{pr}_i$). Since we are applying it to the third quadrant and in particular the axis-order is $x_2, x_1$, we must swap the definitions of $s_1$ and $s_2$ (i.e., $s_1$ will be equal to the difference in the $x_2$ coordinate of $p$ and $q$).

For the pair $p,q$ we have $s_1=1$, $s_2=0$. Thus the left interval consists of the singleton $[\lambda,\lambda]$, the right interval is $[\lambda+1, \lambda+n-1]$, $X^{pq}_3(n)=\emptyset$ and we are splitting the numbers of the right interval into sets $X^{pq}_1(n)$ and $X^{pq}_2(n)$ depending on whether or not they are larger than $\lambda$. That is, 
\begin{eqnarray*}
X^{pq}_1(n)=[\lambda+1, \lambda+n-1] \cap \{i\in\Z \colon i\prec_{\theta^p_{\mat{t}'}} \lambda\} \\
X^{pq}_2(n)=[\lambda+1, \lambda+n-1] \cap \{i\in\Z \colon \lambda\prec_{\theta^p_{\mat{t}'}} i\} 
\end{eqnarray*}
Applying the subtree property to the pair $p,r$ gives a similar partition. In this case, the three sets become $X^{pr}_1(n)=\emptyset$, $X^{pr}_2(n)=[\lambda+1, \lambda+n-1] \cap \{i\in\Z \colon i\prec_{\theta^p_{\mat{t}'}} \lambda\}=X^{pq}_1(n)$, and  $X^{pr}_3(n)=[\lambda+1, \lambda+n-1] \cap \{i\in\Z \colon \lambda\prec_{\theta^p_{\mat{t}'}} i\}=X^{pq}_2(n)$.

The sets $X_i^{pq}$ imply some constraints on $\theta^q_{\mat{t}'}$ (similarly, $X_i^{pr}$ gives constraints on $\theta^r_{\mat{t}'}$). 
Recall that we previously observed that $\theta^q_{\mat{t}'}[\lambda+1, \infty)=\theta^p_{\mat{t}'}[\lambda+1, \infty)=(\theta+2)[\lambda+1, \infty)$, which in particular implies that all constraints of the subtree property apply to $\theta+2$.

$X^{pq}_2(n)$ says that all relationships in $\theta^p_{\mat{t}'}[\lambda+1,\lambda+n-1]$ are be preserved for numbers that are larger than $\lambda$ in $\theta^p_{\mat{t}'}$. Similarly, $X^{pr}_2(n)$ says that relationships for numbers {\em smaller} than $\lambda$ must also be preserved. Thus, we conclude that all relationships (both larger and smaller than $\lambda$) must be preserved. Hence, we conclude that $\theta^p_{\mat{t}'}[\lambda+1,\lambda+n-1] \subset (\theta+2)[\lambda+1,\infty)$. This reasoning applies for any values of $\lambda\in \Z$, and $n>0$. In particular, when $\lambda \rightarrow -\infty$ and $n \rightarrow \infty$ we get $\theta = \theta +2$ as claimed.

\end{proof}

With this result we can now show Theorem~\ref{theo_shiftedCDS}.

\begin{proof}[(Proof of Theorem~\ref{theo_shiftedCDS})]
Let $\mat{t}'=(+1, \ldots, +1)$ and note that, by definition, we have $\theta^p_{\mat{t}'}=\theta$. We apply Theorem~\ref{theo:CDR} 
and obtain $\theta^p_{\mat{t}}[ \mat{t} \cdot p, \infty) = \theta^p_{\mat{t}'}[ \mat{t}' \cdot p, \infty) - \mat{t}' \cdot p + \mat{t} \cdot p= \theta[ \mat{t}' \cdot p, \infty) - \mat{t}' \cdot p + \mat{t} \cdot p$. The term $-\mat{t}' \cdot p + \mat{t} \cdot p$ must be an even number (it is the inner product of $p$ with vector $\mat{t}-\mat{t}'$ which satisfies that each coordinate is either a zero or a two). Thus, we can apply $\theta = \theta + 2$ repeatedly until we get $\theta[\mat{t}' \cdot p, \infty) - \mat{t}' \cdot p + \mat{t} \cdot p = \theta[\mat{t} \cdot p, \infty)$ as claimed.
\end{proof}

Specifically, we give two necessary conditions that together are also sufficient. The two conditions are derived from the axioms S1-S5. 
The first necessary condition is $\theta = \theta + 2$, which is already proved in Lemma~\ref{lem:nec_1}.

The other necessary condition derives from the symmetry axiom (S2) of CDSs.

\begin{lemma}[Necessary condition 2 for CDSs]
\label{lem:nec_2}
Any total order such that $\CDS(\theta)$ forms a CDS satisfies that $\theta = -(\theta+1)^{-1}$.
\end{lemma}
\begin{proof}[Proof (sketch)]
This proof follows the same spirit as Theorem~\ref{theo:CDR}, but using the symmetry axiom instead. For any two numbers $a,b$ such that such that  $a \prec_\theta b$ we choose two points $p,q\in \Z^d$ and look at $R(p,q)$. In particular, we look at two specific intermediate points $r$ and $s$. The key property of these two points is that the behaviour of $R(p,q)$ around those points is determined by the positions of $a$ and $b$ in $\theta$. Then, we look at the symmetric path $R(q,p)$ and show that the behaviour around the same intermediate points now depends on the positions of $-b-1$ and $-a-1$. In order to satisfy the symmetry axiom, the return path $R(q,p)$ has to be the same and thus we must have $-b-1 \prec_\theta -a-1$. %

\end{proof}

This completes one side of the implication of Theorem~\ref{theo_characCDS}. In order to complete the proof we show that the  two requirements for $\theta$ are also sufficient. 

\begin{lemma}[Sufficient condition for CDSs]
\label{lem:CDS}
Let $\theta$ be a total order that satisfies $\theta+2 = \theta$ and $\theta = -(\theta+1)^{-1}$. Then, $\CDS(\theta)$ forms a CDS.
\end{lemma}

\section{Characterization of necessary and sufficient conditions}\label{sec_carac}
Let $\mathcal{F}$ be the collection of total orders of $\Z$ that satisfy the necessary and sufficient conditions of Theorem~\ref{theo_characCDS}. 
In order to bound the Hausdorff distance of the CDS associated to these constructions, we must give properties of total orders in $\mathcal{F}$. %

\begin{observation}\label{obs_monot}
All odd numbers appear monotonically in any total order $\theta$ that satisfies $\theta=\theta+2$. The same holds for even numbers.
\end{observation}

The above result follows from repeatedly applying the fact that $a\prec_{\theta} b \Leftrightarrow a+2 \prec_{\theta} b+2$. The second necessary condition also gives a strong relationship between odd and even numbers.

\begin{observation}\label{obs_monotboth}
Let $\theta$ be a total order such that $\theta=-(\theta+1)^{-1}$. Then, it holds that $0\prec_{\theta} 2 \Leftrightarrow -3 \prec_{\theta} -1$. 
\end{observation}

By combining the previous two observations we get that either both odd and even numbers increase monotonically for any $\theta\in \mathcal{F}$ or both decrease monotonically. Next we study the relationship between odd and even numbers.

\begin{lemma}\label{lem_consecutive}
Let $\theta\in\mathcal{F}$ be a total order in which two numbers of the same parity are consecutive in $\theta$. Then, it holds that $1\prec_{\theta} 2 \Leftrightarrow 2q+1 \prec_{\theta} 2q'$ for all $q,q'\in\mathbb{Z}$.
\end{lemma}

\begin{corollary}\label{cor_four}
There are exactly four total orders in $\mathcal{F}$ in which two numbers of the same parity are consecutive.
\end{corollary}
\begin{proof}
Let $\theta\in\mathcal{F}$ be any such total order. By Lemma~\ref{lem_consecutive} either all odd numbers appear before all even numbers or vice versa. There are four cases depending on whether $0\prec_{\theta} 2$ or $2\prec_{\theta} 0$ and $1\prec_{\theta} 2$ or $2\prec_{\theta} 1$. The first inequality determines whether all even numbers appear monotonically increasing or decreasing in $\theta$ (by Observations \ref{obs_monot} and \ref{obs_monotboth} this also determines how do odd numbers appear). The second inequality determines whether odd numbers are smaller or larger (with respect to $\prec_{\theta}$) than the even ones. Thus, under the assumption that two numbers of the same parity are consecutive in $\theta$, only the following four orders exist.
\begin{eqnarray*}
\tau_{o^+e^+} &=&  \{ \ldots \prec 1 \prec 3 \prec 5 \prec \ldots \prec 0\prec 2 \prec 4 \prec \ldots  \}\\
\tau_{o^-e^-} &=&  \{ \ldots \prec 5 \prec 3 \prec 1 \prec \ldots \prec 4\prec 2\prec 0 \prec \ldots  \}\\
\tau_{e^+o^+} &= (\tau_{o^-e^-})^{-1} =&  \{ \ldots \prec 0 \prec 2 \prec 4 \prec \ldots \prec 1\prec 3 \prec 5 \prec \ldots  \}\\
\tau_{e^-o^-} &= (\tau_{o^+e^+})^{-1} =&  \{ \ldots \prec 4 \prec 2 \prec 0 \prec \ldots \prec 5\prec 3 \prec 1 \prec\ldots  \}
\end{eqnarray*}
\end{proof}

It remains to consider the case in which $\theta\in\mathcal{F}$ is a total order in which no two numbers of the same parity appear consecutively. That is, we have an odd number followed by an even number, followed by an odd number, and so on. For any $q\in\Z$, let $\alpha_{q}$ be the total order satisfying $\ldots \prec_{\alpha_{q}} 0 \prec_{\alpha_{q}} 2q+1\prec_{\alpha_{q}} 2\prec_{\alpha_{q}} 2q+3 \prec_{\alpha_{q}} 4 \prec_{\alpha_{q}} \ldots$. 

\begin{theorem}\label{theo_shapef}
$\mathcal{F}= \{\tau_{o^+e^+},\tau_{o^-e^-},\tau_{e^+o^+},\tau_{e^-o^-}\} \cup \{\alpha_{q} \colon q\in \Z\} \cup \{(\alpha_{q})^{-1} \colon q\in \Z\}$
\end{theorem}

This completely characterizes the set $\mathcal{F}$ of total orders, and allows us to find a lower bound on the Hausdorff distance of the associated CDSs. %

\begin{theorem}\label{theo:dist}
For any $p=(p_1, \ldots, p_d)\in\Z^d$, total order $\theta \in \mathcal{F}$ and $n>0$, there exists a point $q\in\Z^d$ such that $||p-q||_2 =  \Theta(n)$ and $H(\overline{pq},R(p,q)) = \Theta(n)$.
\end{theorem}
\begin{proof}[Proof (sketch)]
Pick a point $q$ sufficiently far from $p$ and look at one every two steps in the path $R(p,q)$. The way in which the path behaves will depend on the position of the odd numbers of $\theta$ (or even numbers depending on the parity of the starting point). Since odd and even numbers appear monotonically in $\theta$, the path will do all steps in one direction before moving into a different one. Intuitively speaking, the movements in the odd numbers will form a bounding box and so will the movements in the even numbers (although the interaction between them may not be same). A specific example of such path is given in the Appendix.
\end{proof}

{\bf Remark} Although, asymptotically speaking, our construction has the same Hausdorff distance as the bounding box CDS, it can be seen that the constant is roughly half: the bounding box CDS has an error of $\frac{\sqrt{2}n}{4}\approx 0.3n$ whereas, say, $\CDS(\tau_{o^+e^+})$ has an error of $\frac{\sqrt{5}n}{15} \approx 0.15n$.

\section{Conclusions}\label{sec:concl}
Increasing the dimension from two to three brings a significant change in the associated constraints for creating CDRs and CDSs. Although we have not been able to create a CDS with $o(n)$ Hausdorff distance, we believe that the results presented in this paper provide the first significant step towards this goal. The next natural step would be to consider constructions that apply different total orders to different points of $\mathbb{Z}^d$.

For simplicity in the exposition, we have defined the CDS as the union of CDRs at all points. The construction of Christ {\em et al.}~\cite{ChristJournal12} considers the union of {\em half CDRs} instead (CDRs that are only defined for half of the slopes, such as slopes that satisfy $t_1=+1$). We note that the same result would follow if we use their approach. Indeed, in order to derive the two necessary conditions, we have only looked at two slopes. For simplicity we have used $(+1,\ldots, +1)$ and $(-1,-1,+1, \ldots , +1)$, but the same result follows for any two slopes that differ in two coordinates. Thus, constructing CDSs by gluing half CDRs would result in the same necessary and sufficient constraints. 

Similarly, one could consider using some kind of priority between slopes (say, lexicographical) so that if $p$ and $q$ are in more than one orthant, only the definition of $R(p,q)$ in the lexicographically smallest slope is considered. This removes the dependency between orthants (Theorem~\ref{theo:CDR}), but has a consistency problem: we can find three points $p,q,q'\in\Z^d$ such that $R(p,q)$ and $R(p,q')$ have different slopes, but the intersection of the two segments is not connected (such as in Figure~\ref{CDR_counterexample}).

\section*{Acknowledgements}
The authors would like to Thank Takeshi Tokuyama and Matthew Gibson for their valuable comments during the creation of this paper.

\begin{figure}
\centering
\scalebox{0.6}{\includegraphics{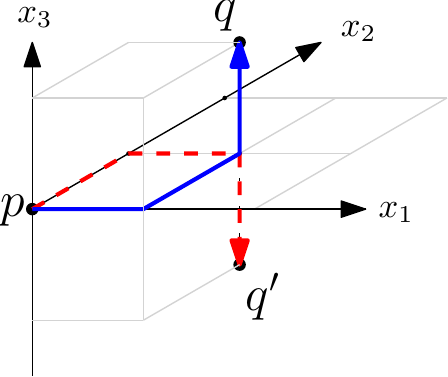}}
\caption{Removing dependency between orthants can creates inconsistencies between each other.
}
\label{CDR_counterexample}
\end{figure}

\bibliographystyle{plainurl}
\bibliography{reference}

\newpage
\appendix
\iffalse
\section{Interesting Figures}

\begin{figure}
\centering
\scalebox{0.6}{\includegraphics{figure/CDR}}
\caption{An example of the CDR in a 2-dimensional subspace at $p$ with the total orders $\{\theta^t : t \in T\}$ that satisfies $\theta_\mat{t}[\mat{t}\cdot p, \infty) = \theta_{\mat{t}'}[ \mat{t}' \cdot p, \infty) - \mat{t}' \cdot p + \mat{t}\cdot p$. Note that each quadrant is a symmetric or rotated copy of each other (this difference is due to changes in the axis-order for each quadrant). For clarity, two portions of the tree are depicted in solid blue and dashed red. %
}
\label{fig:exampleCDR}
\end{figure}
\mati{What do we say about the figure?}
\fi
\section{Proofs omitted from the main document}
\section*{Proof of Lemma~\ref{lem:DO_tran}}
\begin{proof}
Since we are doing a translation operation, the segment from $o$ to $q-p$ has the same slope $\mat{t}$ as the segment from $p$ to $q$. Thus, we must compare $\mathcal{C}_p=\CDO(\theta, p, \mat{t})$ and $\mathcal{C}_o=\CDO(\theta- \mat{t}\cdot p, o, \mat{t})$.
Let $p=m_1, \ldots m_k=q$ be the path from $p$ to $q$ in $\mathcal{C}_p$, and let 
$o=w_1, \ldots w_k=q-p$ be the path from $o$ to $q-p$ in $\mathcal{C}_o$. We show that $m_{i}-p= w_{i}$ for all $i<k$ by induction.

The base case $m_1-p=w_1$ follows from definition of $p$. So assume that this property holds for some $i<k$. The decision of which direction to move in $\mathcal{C}_p$ depends on the position of $\mat{t}\cdot m_i$ in $\theta[\mat{t}\cdot p, \mat{t}\cdot q-1]$. Similarly, in $\mathcal{C}_o$, this choice depends on the position of $\mat{t}\cdot w_i=\mat{t}\cdot m_i-\mat{t}\cdot p$ in $(\theta- \mat{t}\cdot p)[\mat{t}\cdot o, \mat{t}\cdot (q-p)-1]=(\theta- \mat{t}\cdot p)[0, \mat{t}\cdot q-1 -\mat{t}\cdot p]$. 

That is, in one case we are looking at the relative position of number $\mat{t}\cdot w_i$ in a total order $\theta$. In the other case we are looking at a number that is $\mat{t}\cdot p$ units smaller in a permutation $\theta'$ that is equal to $\theta$ where everything has also been reduced by $\mat{t}\cdot p$ (even the scope of our interest). Thus, the relative positions are the same, which implies that both $m_i$ and $w_i$ move in the same direction. Hence, we have $m_{i+1}-p=w_{i+1}$ and by induction the two paths are the translation of each other.
\end{proof}

\section*{Proof of Proposition~\ref{prop_nocdr}}
\begin{proof}
For any slope $\mat{t} = (t_1, t_2, \ldots, t_d)$, let $a_1, \ldots a_k$ be indices of the coordinates with positive value in increasing order (that is, $t_i=+1$ if and only if $i=a_j$ for some $j\leq k$). Similarly, let $b_1, \ldots b_{d-k}$ be the indices of the the coordinates with negative value in decreasing order. By the definition, the axis-order of $(t_1, t_2, \ldots, t_d)$ is $x_{a_1}, x_{a_2}, \ldots, x_{a_k}, x_{b_1}, \ldots, x_{b_{d-k}}$. Let $q$ be a point in the orthant of $\mat{t}$ at $p$ so that $q_i \geq p_i$ for $i=a_1, \ldots, a_k$ and $p_i \geq q_i$ for $i = b_1, \ldots, b_{d-k}$.

Consider the total order construction. If $x_1 = x_{a_1}$ (that is, $p$ has smaller or equal $x_1$ coordinate than $q$), the path moves from $p$ to $(q_1, p_2, \ldots, p_d)$ because $\mat{t} \cdot p$, \ldots, $\mat{t} \cdot p + |q_1 - p_1| -1$ are the $|q_1 - p_1|$ smallest elements in $\theta_0[\mat{t}\cdot p, \mat{t}\cdot q-1]$. Otherwise, $x_1 = x_{b_{d-k}}$ and the last movements will be in the $x_1$ coordinate (i.e., $\mat{t} \cdot q - |p_1-q_1|, \ldots, \mat{t} \cdot q -1$ are the $|p_1-q_1|$ largest elements in $\theta_0[\mat{t}\cdot p, \mat{t}\cdot q-1]$). In this case, the path will move from $(p_1, q_2, \ldots, q_d)$ to $q$. Similarly, for $i = 2, \ldots d$, the path extends $|q_i-p_i|$ steps in $x_i$ coordinate and is connected from $p$ or $q$ depending on $q_i \geq p_i$ or $p_i \geq q_i$. This is the same as the bounding box approach.

For the second part of the claim, we consider the case in which $d=3$, and $p=(0,0,2)$. Let $q=(1,1,3)$ and $q'=(1,1,1)$; we claim that paths $R(p,q)$ and $R(p,q')$ do not satisfy the subsegment property. 

The slope of $R(p,q)$ is $(+1, +1, +1)$ while the slope of $R(p,q')$ is $(+1, +1, -1)$.
Indeed, $R(p,q)$ is constructed by $\theta_1[2, 4] =\{ 2\prec 3\prec 4 \}$ with axis-order $x_1, x_2, x_3$, so the path is $(0,0,2) \rightarrow (1,0,2) \rightarrow (1,1,2) \rightarrow (1,1,3)$. On the other hand, $R(p,q')$ is constructed by $\theta_1[-2,0] = \{ -1 \prec -2 \prec 0\}$ with axis-order $x_1, x_2, x_3$, so the path is $(0,0,2) \rightarrow (0,1,2) \rightarrow (1,1,2) \rightarrow (1,1,1)$. In particular, notice that $R(p,q) \cap R(p,q') = \{(0,0,2), (1,1,2)\}$ which is not connected under the 8-neighbor topology. i.e. the path from $(0,0,2)$ to $(1,1,2)$ is not well-defined (See Figure~\ref{CDR_counterexample}).
\end{proof}

\section*{Proof of Lemma~\ref{lem:necCDR}}
\begin{proof}
We will prove the statement by contradiction: assume that there exist two slopes $\mat{t}, \mat{t}'$ such that $v \prec_{\theta_\mat{t}} v' $ but $v' - \mat{t} \cdot p + \mat{t}'\cdot p \prec_{\theta_{\mat{t}'}} v - \mat{t} \cdot p + \mat{t}'\cdot p $. We will pick a point $q$ such that $R(p,q)$ has both slope $\mat{t}$ and $\mat{t}'$ (i.e., it is in the intersection of both orthants), and look at $R(p,q)$ from both the viewpoints of $\CDO(\theta_\mat{t}, p, \mat{t})$ and $\CDO(\theta_{\mat{t}'}, p, \mat{t}')$. Along the path $R(p,q)$ we look at two intermediate points; the way in which the path behaves at those intermediate points will depend on the positions of $v$ and $v'$ in $\theta_\mat{t}$ (from the viewpoint of $\CDO(\theta_\mat{t}, p, \mat{t})$). Similarly, the behaviour of the same path on the other orthant will depend on the positions of $v' - \mat{t} \cdot p + \mat{t}'\cdot p$ and $v - \mat{t} \cdot p + \mat{t}'\cdot p $ in $\theta_{\mat{t}'}$. Thus, if the relationships are reversed, the two paths will behave differently, and in particular we cannot have a CDR.

More formally, assume that $\theta_\mat{t}[\mat{t}\cdot p, \infty) \neq \theta_{\mat{t}'}[ \mat{t}' \cdot p, \infty) - \mat{t}' \cdot p + \mat{t}\cdot p$. That is, there exist $v,v'\geq \mat{t}\cdot p$ such that $v \prec_{\theta_\mat{t}} v' $ but $v' - \mat{t} \cdot p + \mat{t}'\cdot p \prec_{\theta_{\mat{t}'}} v - \mat{t} \cdot p + \mat{t}'\cdot p $. For simplicity in the exposition, we will first consider the case in which $\mat{t}, \mat{t}'$ that only differ in one coordinate. Since $\mat{t}$ and $\mat{t}'$ differ in a single coordinate, the intersection of the two associated orthants forms a subspace $\mathcal{H}$ of dimension $d-1 \geq 2$. By the definition of the total order construction, from any point $q\in \mathcal{H}$ we can construct $R(p,q)$ either from $\CDO(\theta_\mat{t}, p, \mat{t})$ or $\CDO(\theta_{\mat{t}'}, p, \mat{t}')$. %

Let $x_c$ be the coordinate that $\mat{t}$ and $\mat{t}'$ differ, and let $x_a$ and $x_b$ be two coordinates that $\mat{t}$ and $\mat{t}'$ have the same value (note that $a$ and $b$ must exist because the dimension of $\mathcal{H}$ is at least two). Without loss of generality we assume that $t_c=+1$, $t'_c=-1$ and $x_a$  precedes $x_b$ in the axis-order of $\mat{t}$.

Let $u = \max\{v,v'\} + 1$ and $i_v$ be the position of $v$ in $\theta_\mat{t}[\mat{t}\cdot p, u-1]$. Starting from $p$ move $t_a i_v$ positions in the $x_a$ coordinate and $t_b(u - \mat{t}\cdot p - i_v)$ positions in the $x_b$ coordinate. Let $q$ be such point (more formally, $q_a=p_a+ t_a i_v$, $q_b=p_b+ t_b(u - \mat{t}\cdot p - i_v)$ and $q_i=p_i$ for all $i\neq a, b$). Figure~\ref{fig:CDR} shows an example of the construction of $q$, $r$ and $r'$.

\begin{figure}[h]
\centering
\includegraphics{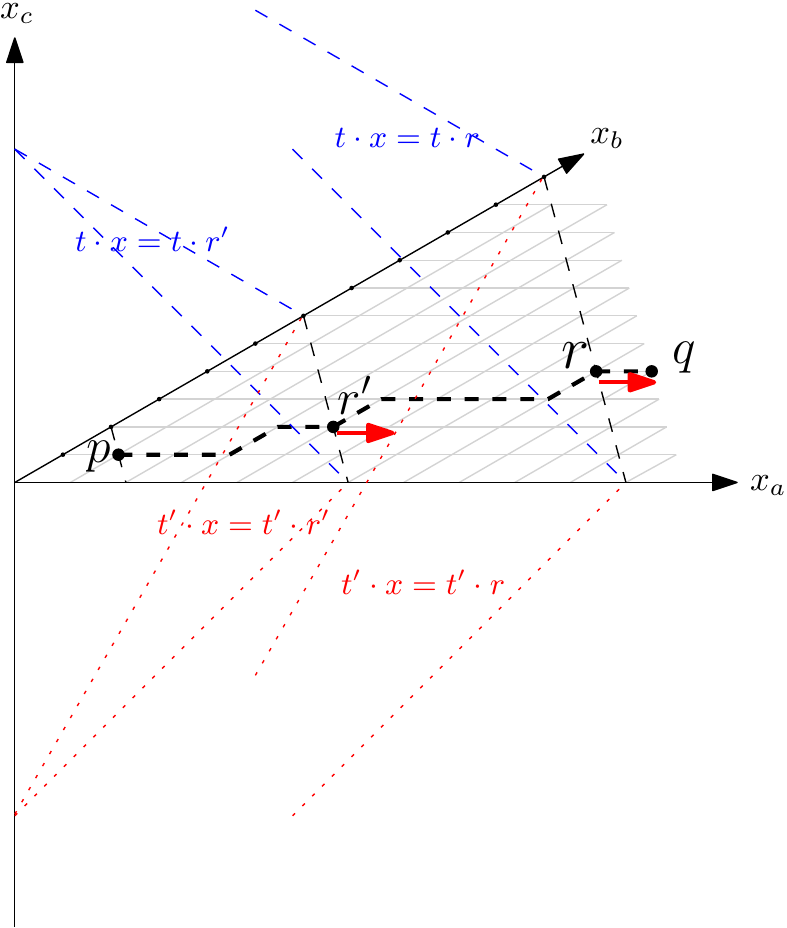}
\caption{Illustration of proof of Lemma~\ref{lem:necCDR}. In the path from $p$ to $q$ we pass through two intermediate points $r$ and $r'$. Their behavior on the two orthants depends on the relative positions of two different pairs of numbers (the constant in which they are swept by the planes $\mat{t} \cdot x=c$ and $\mat{t}'\cdot x = c$, respectively).
}
\label{fig:CDR}
\end{figure}

By construction, point $q$ lies in the orthants associated to both $\mat{t}$ and $\mat{t}'$ (and thus $q\in\mathcal{H}$). Moreover, we have $\mat{t} \cdot q = u$.  We compare the construction of $R(p,q)$ using both $\CDO(\theta_\mat{t}, p, \mat{t})$ and $\CDO(\theta_{\mat{t}'}, p, \mat{t}')$ and observe that indeed they do not match. From the viewpoint of $\CDO(\theta_\mat{t}, p, \mat{t})$, $R(p,q)$ is constructed by $\theta_\mat{t}[\mat{t}\cdot p, \mat{t}\cdot q-1 = u-1]$ with axis-order $x_a,x_b$. We only need to consider $x_a,x_b$ in the axis-order because the path only moves $i_v$ steps in the $x_a$ coordinate and $u - \mat{t}\cdot p - i_v$ steps in the $x_b$ coordinate. Recall that $\mat{t}\cdot p \leq v,v' \leq u-1$ and $\mat{t} \cdot q = u =\max\{v,v'\} + 1$. Thus, along the path $R(p,q)$ we must traverse through two points $r,r' \in R(p,q)$ such that $\mat{t} \cdot r = v$ and $\mat{t} \cdot r' = v'$. %
 We now consider the movements of $R(p,q)$ at the two intermediate points.

By the total order construction, this will depend on the positions of $v$ and $v'$ in $\theta_\mat{t}[\mat{t}\cdot p, u-1]$: the $i_v$ smallest elements will correspond to points in which we move in the $x_a$ coordinate (and in the  remaining $u - \mat{t}\cdot p - i_v$ cases we move in the $x_b$ coordinate). By construction, $v$ is the $i_v$-th element in $\theta_\mat{t}[\mat{t}\cdot p, u-1]$, which in particular implies that at point $r$ the path will move in the $x_a$ coordinate. On the other hand, the path at $r'$ moves in the $x_b$ coordinate because $v \prec_{\theta_\mat{t}} v' $ (and thus $v'$ is not among the $i_v$ smallest elements in $\theta_\mat{t}[\mat{t}\cdot p, u-1]$).

Now we look at $R(p,q)$ from the viewpoint of $\CDO(\theta_{\mat{t}'}, p, \mat{t}')$. In this case, the path $R(p,q)$ is constructed by $\theta_{\mat{t}'}[\mat{t}'\cdot p, \mat{t}'\cdot q-1]$ with axis-order $x_a,x_b$ (as before, we only need to look at the axis-order of these two coordinates because the path stays within that subspace). The path must contain points $r$ and $r'$ or it would not be consistent. Next, we consider the movements of the path at these two points. How we move on these two points will now depend on the positions of $\mat{t}' \cdot r$ and $\mat{t}' \cdot r'$ in $\theta_{\mat{t}'}[\mat{t}'\cdot p, \mat{t}'\cdot q-1]$ (again, the $i_v$ smallest elements will correspond to points in which we move in the $x_a$ coordinate).

Recall that we assumed that $v' - \mat{t} \cdot p + \mat{t}'\cdot p \prec_{\theta_{\mat{t}'}} v - \mat{t} \cdot p + \mat{t}'\cdot p$. We simplify the first term:

\begin{eqnarray*}
v' - \mat{t} \cdot p + \mat{t}'\cdot p &=& \mat{t} \cdot r'  - \mat{t} \cdot p + \mat{t}'\cdot p\\
 &=& \mat{t} \cdot r' - 2p_c \\ 
 &=& \mat{t}' \cdot r' + \mat{t} \cdot r' - \mat{t}' \cdot r' - 2 p_c \\
 &=& \mat{t}' \cdot r' + 2 r_c - 2 p_c \\
 &=& \mat{t}' \cdot r'
\end{eqnarray*}

Where the first equality follows from the definition of $r'$, and the second and third follow from the fact that $\mat{t}$ and $\mat{t}'$ only defer in the $c$-th coordinate, and the last one from the fact that $r_c=p_c$. Similarly, we can show that $v - \mat{t} \cdot p + \mat{t}'\cdot p= \mat{t}'\cdot r$, and thus $v' - \mat{t} \cdot p + \mat{t}'\cdot p \prec_{\theta_{\mat{t}'}} v - \mat{t} \cdot p + \mat{t}'\cdot p$ simplifies to $\mat{t}' \cdot r' \prec_{\theta_{\mat{t}'}} \mat{t}' \cdot r$. 

We use this result to argue about the positions of $\mat{t}' \cdot r$ and $\mat{t}' \cdot r'$ in $\theta_{\mat{t}'}[\mat{t}'\cdot p, \mat{t}'\cdot q-1]$.
In order to be consistent with the construction in the orthant associated to $\mat{t}$, $\mat{t}' \cdot r$ must be in the $i_v$ smallest elements in $\theta_{\mat{t}'}[\mat{t}'\cdot p, \mat{t}'\cdot q-1]$. However, this implies that $\mat{t}' \cdot r'$ also is among the $i_v$ smallest elements in $\theta_{\mat{t}'}[\mat{t}'\cdot p, \mat{t}'\cdot q-1]$. Thus, the path at both ${r}$ and ${r}'$ moves in the $x_a$ coordinate. This contradicts with $R(p,q)$ in $\CDO(\theta_\mat{t}, p, \mat{t})$.

In order to complete the proof we must consider the case in which $\mat{t}$ and $\mat{t}'$ differ in more than one coordinate. In this case we pick a sequence of orthants connecting $\mat{t}$ and $\mat{t}'$ in a way that two consecutive orthants only differ in a single coordinate. If $\theta_\mat{t}[\mat{t}\cdot p, \infty) \neq \theta_{\mat{t}'}[ \mat{t}' \cdot p, \infty) - \mat{t}' \cdot p + \mat{t}\cdot p$, then there must exist two orthants that are consecutive in the sequence and do not satisfy the equality either. This situation would contradict with the previous reasoning. Thus, we conclude that this cannot happen which completes the proof.
\end{proof}

\section*{Proof of Lemma~\ref{lem_sufficient}}
\begin{proof}
First notice that it suffices to show that the construction is well defined. Indeed, if so Lemma~\ref{lem:CDO} implies that each orthant on its own is a partial CDS (and in particular satisfies the axioms). Segments of different slopes can only meet at the intersection of both associated orthants, and by monotonicity, the union must also satisfy all five axioms. 

Thus, we must show that for any $q\in\Z^d$ there is a unique path $R(p,q)$. If $q$ belongs to a single orthant defined by a slope $\mat{t}$, then the path is unique (since it is only defined in $\CDO(\theta_{\mat{t}},p,\mat{t})$). Thus, we study the case in which a point belongs to two (or more) orthants.

For any $\mat{t} \neq \mat{t}' \in T$, the intersection of the two associated orthants will form a lower dimensional subspace $\mathcal{H}$ that is in the boundary of both orthants. To complete the proof, we need to show that the set of segments on $\mathcal{H}$ are also consistent. Using the translation lemma (Lemma~\ref{lem:DO_tran}), we look at the translated version at the origin $o$ instead.

Let $d'$ be the dimension of $\mathcal{H}$, and let $x_{a_1}, x_{a_2}, \ldots, x_{a_{d'}}$ be the coordinates that span $\mathcal{H}$.
Note that, $\mat{t}$ and $\mat{t}'$ have the same values in the $x_{a_1}, x_{a_2}, \ldots, x_{a_{d'}}$ coordinates. By Observation~\ref{obs:same_order}, both $\mat{t}$ and $\mat{t}'$ have the same axis-order $\tau$ restricted to $\mathcal{H}$. The segments of $\CDO(\theta_{\mat{t}} - \mat{t} \cdot p, o,\mat{t})$ on $\mathcal{H}$ are constructed by $\theta_{\mat{t}}[\mat{t} \cdot p, \infty) - \mat{t} \cdot p$ with axis-order $\tau$. On the other hand, the segments of $\CDO(\theta_{\mat{t}'} - \mat{t}' \cdot p, {o},\mat{t}')$ on $\mathcal{H}$ are constructed by $\theta_{\mat{t}'}[\mat{t}' \cdot p, \infty) - \mat{t}' \cdot p$ with the same axis-order $\tau$. Since $\theta_\mat{t}[\mat{t}\cdot p, \infty) = \theta_{\mat{t}'}[ \mat{t}' \cdot p, \infty) - \mat{t}' \cdot p + \mat{t} \cdot p$, the two sets of segments on $\mathcal{H}$ must be the same and thus the path is well defined as claimed.
\end{proof}

\section*{Proof of Lemma~\ref{lem:prefix}}
\begin{proof}
We now show that if $i_a \leq r_1 - p_1 \leq i_b -1$, $R(p,r)$ passes through $q$. 
 We will show that $R(p,q)$ and $R(p,r)$ behave in the same way before reaching the diagonal line $x_1 + x_2 = q_1 + q_2$. Since $r_1\geq q_1$ and  $r_2 \geq q_2$, both $q$ and $r$ are above and to the right of $p$ (and thus the slope of both segments is $\mat{t} = (+1,+1)$).

We distinguish three cases of $s_1, s_2$ separately. First consider the case in which both $s_1, s_2>0$. By the total order construction, the movements of $R(p,r)$ before reaching the diagonal line depend on the positions of the numbers $[\mat{t} \cdot p, \mat{t} \cdot q - 1]$ in $\theta^p_\mat{t}[\mat{t} \cdot p, \mat{t} \cdot r - 1]$. The $r_1-p_1$ smallest elements will correspond to points in which we move in the $x_1$ coordinate (and in the remaining $n-r_1+p_1$ ones in the $x_2$ coordinate).
Since there are at least $i_a$ and at most $i_b-1$ horizontal movements from $p$ to $r$, $a$ must be among the $r_1-p_1$ smallest elements in the complete order $\theta^p_\mat{t}[\mat{t} \cdot p, \mat{t} \cdot r - 1]$. Similarly,  $b$ must be among the $n-r_1+p_1$ largest elements. %

Recall that, by the definition of $a$ and $b$, these two numbers correspond to the last movement in the $x_1$-axis and the first movement in the $x_2$-axis in the left order. Thus, all $s_1$ numbers that correspond to movements in the $x_1$-axis in $R(p,q)$ are smaller than or equal to $a$ in $\theta^p_\mat{t}$ (similarly, the $s_2$ numbers that correspond to movements in the $x_2$-axis are larger than or equal to $b$). When we consider these numbers in the complete order, these $s_1$ numbers will belong to the $r_1-p_1$ smallest elements (and the last $s_2$ numbers in the left order, respectively). In particular, the same movements will be done until we reach $q$. 

Now, we consider the case $s_1=0$ (recall that in this case we have $i_a=0$). With a similar reasoning we observe that $b$ is among the $n-r_1+p_1$ largest elements in the complete order. Again, by the definition of $b$, we have that $b$ is the smallest elements in the left order. In particular, all numbers in the left order are larger than or equal to $b$ and thus in the $n-r_1+p_1$ largest elements in the complete order (and thus correspond to movements in the $x_2$-axis). The case in which $s_2=0$ is symmetric. This proves that $R(p,q) \subset R(p,r)$.

We now consider the reverse statement. By the monotonicity axiom, if $R(p,r)$ passes through $q$, $r$ must be above and to the right of $q$ (or it would violate the monotonicity axiom). Assume, for the sake of contradiction, that $R(p,r)$ passes through $q$, but $r_1-p_1 < i_a$ or $r_1-p_1 > i_b-1$. Axiom $(S3)$ implies that $R(p,q) \subset R(p,r)$; we will see that the two paths $R(p,q)$ and $R(p,r)$ will split before reaching $q$, giving a contradiction. 

First consider the case $r_1-p_1 < i_a$. In this case, $a$ belongs to the $n-r_1+p_1$ largest elements in the complete order. In particular, $R(p,r)$ will make a vertical movement at the intermediate point $q'$ such that $\mat{t} \cdot q' = a$. However, by the definition of $a$ in the left order, $R(p,q)$ makes a horizontal movement at $q'$ (recall that $a$ is the largest number for which $R(p,q)$ makes a movement in the $x_1$ coordinate). 

This implies that the two paths diverge and gives a contradiction as claimed. The case $r_1-p_1 > i_b-1$ is analogous. In this case, the movements will diverge at the intermediate point $q''$ such that $\mat{t} \cdot q'' = b$.%
\end{proof}

\section*{Proof of Lemma~\ref{lem_subtree}}
\begin{proof}
Before giving the proof, we note that this result follows from Lemma~1 of~\cite{Chowdhury2016}. Their statement is slightly more general and uses several results on a {\em contract} operation (defined in~\cite{Chowdhury2016}) that CDSs must satisfy. In the following, we give an alternative proof that is based on geometric properties and does not need the contracting machinery.

By Lemma~\ref{lem:prefix}, $i_a$ and $i_b$ partition the complete order into 3 parts: the first $i_a$ numbers, the last $n-i_b+1$ numbers and the remaining numbers. In order to pass through $q$, the first $i_a$ numbers in the complete order will correspond to points in which we move in the $x_1$ coordinate and the last $n-i_b+1$ numbers in the complete order will correspond to points in which we move in the $x_2$ coordinate. Since these numbers have clear movements, the order within themselves is not important. The numbers in the positions from $i_a+1$ to $i_b-1$ are larger than $\mat{t} \cdot q - 1$, which will determine the movements from $q$, so we need to keep the same order as those in $\theta^p_\mat{t}$.
Since we are only interested in the path starting from $q$, we can remove the numbers from $\mat{t} \cdot p$ to $\mat{t} \cdot q -1$ in these three partitions in the complete order and it will give us $X_1(n)$, $X_2(n)$ and $X_3(n)$.
\end{proof}

\section*{Proof of Lemma~\ref{lem:nec_2}}
\begin{proof}

Let $a,b$ be any two numbers such that  $a \prec_\theta b$. We will show that $-b-1 \prec_\theta -a-1$.  Let $c = \min\{a,b\}$ and $d =\max\{a,b\}+1$, Let $i_a$ be the position of $a$ in $\theta[c,d-1]$ and $\mat{t} = (+1,+1,+1, \ldots)$. Consider now the path between $p=(0,c,0,\ldots, 0)$ and $q=(i_a,d-i_a,0,\ldots,0)$. %
 By construction we have $\mat{t} \cdot p = c$, $\mat{t} \cdot q = d$ and $q_2 - p_2 = d-i_a-c \geq 0$ because $1 \leq i_a \leq d-c$.

By definition, path $R(p,q)$ is constructed by $\theta[c,d-1]$ with axis-order $x_1,x_2$. We only need to consider $x_1,x_2$ in the axis-order because the path only moves $i_a$ steps in the $x_1$ coordinate, $d - c - i_a$ steps in the $x_2$ coordinate and does not move in other dimensions. Note that there must exist two intermediate points $r, s \in R(p,q)$ such that $r_1+r_2 = a$ and $s_1+s_2 = b$.

We study the path $R(p,q)$ around points $r$ and $s$. Since the position of $a$ in $\theta[c,d-1]$ is $i_a$, the path at point $r$ moves in the $x_1$ coordinate. Recall that $a \prec_{\theta} b$, which implies that $b$ is not among the $i_a$ smallest elements of $\theta[c,d-1]$. Thus, the path at point $s$ moves in the $x_2$ coordinate. Note that in particular, this implies that points $r' = {r} + (1,0,0,\ldots)$ and $s' = s + (0,1,0,\ldots)$ both are in $R(p,q)$ (and $R(q,p)$ or it would violate axiom (S2)).

Then, we consider the reverse path $R(q,p)$, which is constructed by $\theta[ -\mat{t}\cdot q, -\mat{t}\cdot p-1]$ and axis-order $x_2,x_1$ because of Observation~\ref{obs:same_order}. We now look at the path $R(q,p)$ around points $r'$ and $s'$. Note that the slope of $R(q,p)$ is $-\mat{t}$, $-\mat{t}\cdot r' = -a-1$ and $-\mat{t}\cdot s' = -b-1$.

By the symmetry property (S2), we know that the path at $r'$ with $-\mat{t}\cdot r' = -a-1$ must move in the $x_1$ coordinate towards $r$ (which implies that $-a-1$ is among the $i_a$ largest elements of $\theta[ -\mat{t}\cdot q, -\mat{t}\cdot p-1]$). Similarly, at point $s'$ we move in the $x_2$ coordinate towards $s$ (and $-b-1$ is among the $d-c-i_a$ smallest elements of $\theta[ -\mat{t}\cdot q, -\mat{t}\cdot p-1]$). In particular $-b-1$ must be smaller than $-a-1$ in $\theta$ as claimed.
\end{proof}

\section*{Proof of Lemma~\ref{lem:CDS}}

\begin{proof}
Recall that $\CDS(\theta)$ is defined as a union of CDRs. In particular, any segment of $\CDS(\theta)$ will satisfy axioms (S1), (S3), (S4) and (S5). We focus on the remaining symmetry axiom.

Let $\mat{t}$ be the slope of $R(p,q)$. %
Let $p=m_0, \ldots m_k=q$ be the path from $p$ to $q$ in $\CDS(\theta, p, \mat{t})$, and let 
$q=w_k, \ldots w_0=p$ be the path from $q$ to $p$ in $\CDS(\theta, q, -\mat{t})$. Note that $m_j$ and $w_j$ lie on the same hyperplane $\mat{t} \cdot x = \mat{t} \cdot p + j$. i.e. $\mat{t} \cdot m_j = \mat{t} \cdot w_j$.
We show that $m_{j+1} - m_{j}  = w_{j+1} - w_{j}$ for all $j<k$ and thus $R(p,q)=R(q,p)$.

Let $\tau$ be the axis-order of $\mat{t}$. By the total order construction, the path  $R(p,q)$ is generated by $\theta[\mat{t} \cdot p, \mat{t} \cdot q -1]$ %
and axis order $\tau$. %
Similarly, $R(q,p)$ is generated by $\theta[-\mat{t} \cdot q, -\mat{t} \cdot p - 1]$. Since we are looking at slope $-\mat{t}$, the axis order is the reverse of $\tau$. Note that reversing both the total order and axis order generates the same path, so equivalently, $R(q,p)$ is generated by the total order $\theta^{-1}[-\mat{t} \cdot q, -\mat{t} \cdot p - 1]$ and axis order $\tau$. %
We then apply the condition  $\theta = -(\theta+1)^{-1}$ and obtain $\theta^{-1}[-\mat{t} \cdot q, -\mat{t} \cdot p - 1] = (-(\theta+1))[-\mat{t} \cdot q, -\mat{t} \cdot p - 1] = - ( (\theta+1)[ \mat{t} \cdot p + 1, \mat{t} \cdot q]) = -( \theta[\mat{t} \cdot p, \mat{t} \cdot q - 1] + 1)$. Hence, $R(q,p)$ is generated by $-( \theta[\mat{t} \cdot p, \mat{t} \cdot q - 1] + 1)$ and axis-order $\tau$. 

Specifically, the behaviour of segment $m_{j+1} - m_{j}$ in $R(p,q)$ depends on the position of $\mat{t} \cdot m_{j}$ in $\theta[\mat{t} \cdot p, \mat{t} \cdot q -1]$, whereas the behaviour of segment $w_{j} - w_{j+1}$ depends on the position of $-\mat{t} \cdot w_{j+1} = -\mat{t} \cdot m_j - 1$ in $-( \theta[\mat{t} \cdot p, \mat{t} \cdot q - 1] + 1)$. Clearly, the relative positions of these two numbers are the same. Since we are using the same axis-order, they will induce the same partition and thus the movements will correspond with each other. %
\end{proof}

\section*{Proof of Lemma~\ref{lem_consecutive}}
\begin{proof}
The reverse implication ($``\Leftarrow"$) is direct, since the right side condition is the particular case in which $q=0$, so we focus in the ``forward" implication ($``\Rightarrow"$). We prove this case by contradiction. 

Assume that $1\prec_{\theta} 2$, there exist two numbers $a,b$ of the same parity such that $a<b$ that appear consecutively in $\theta$, but there exist some $q,q' \in \Z$ such that $2q' \prec_{\theta} 2q + 1$. There are four cases to consider depending on the parity of $a$ (odd or even) and majority with respect to $\prec_{\theta}$ ($a \prec_\theta b$ or $b \prec_\theta a$). %
We prove the case of $a=2k \prec_\theta b=2k+2$. The proofs for the three other cases ($2k + 2 \prec_\theta 2k$, $2k+1 \prec_\theta 2k+3$ and $2k+3 \prec_\theta 2k+1$) are similar.

Recall that $\theta$ is a total order in $\mathcal{F}$ and in particular it satisfies $\theta=\theta+2$. We apply this equation $|k-q'+1|$ times on $2q' \prec_\theta 2q+1$ and obtain $2k + 2 \prec_\theta 2(q + k - q') + 3$, which in particular implies that there exists some odd number that is larger than $2k + 2$ in $\theta$.
Let $c$ be the first odd number that appears after $2k+2$ in $\theta$ (that is, any odd number $c'$ such that $c' \prec_\theta c$ satisfies $c' \prec_\theta 2k+2$). Since $2k \prec_\theta 2k + 2$ that means even numbers increase monotonically in $\theta$, odd numbers also increase monontonically in $\theta$, and thus $c-2 \prec_\theta c$.
Since there is no number between $2k$ and $2k+2$ in $\theta$ and $c$ is the first odd number after $2k+2$ in $\theta$, we must have 
$c-2 \prec_\theta 2k \prec_\theta 2k+2 \prec_\theta c$. 

Then, we apply $\theta=\theta +2$ to $c-2 \prec_\theta 2k$ and obtain $c\prec_\theta 2k+2$, which causes a contradiction with $2k+2 \prec_\theta c$. This contradicts with the initial assumption of two numbers of the same parity appearing consecutively in $\theta$.

\end{proof}

\section*{Proof of Theorem~\ref{theo_shapef}}
\begin{proof}
It is straightforward to verify that both $\alpha_{q}$ and $(\alpha_{q})^{-1}$ are total orders in $\mathcal{F}$ for all values of $q$ (i.e., we need to verify that they satisfy the two necessary and sufficient conditons). By corollary~\ref{cor_four}, we also know $\{\tau_{o^+e^+},\tau_{o^-e^-},\tau_{e^+o^+},\tau_{e^-o^-}\}$ are also in $\mathcal{F}$.

Thus, it suffices to show that for any total order $\theta$ such that $\theta=\theta+2$ and $\theta=-(\theta+1)^{-1}$, but $\theta\not\in\{\tau_{o^+e^+},\tau_{o^-e^-},\tau_{e^+o^+},\tau_{e^-o^-}\}$ there exists $q\in\Z$ such that $\theta=\alpha_{q}$ or $\theta=\alpha_{q}^{-1}$. Specifically, we claim that if $0\prec_{\theta} 2$ then $\theta=\alpha_{q}$ (for some $q\in\Z$ that will be specified later). Otherwise, we have $2\prec_{\theta} 0$ and $\theta=(\alpha_{q})^{-1}$ instead. Consider first the $0\prec_{\theta} 2$ case. By Observations~\ref{obs_monot} and~\ref{obs_monotboth}, we know that two numbers $a,b\in\Z$ of the same parity will satisfy $a\prec_{\theta} b$ if and only if $a<b$. 

Now we focus in the relationship between two numbers of different parities. Assume, without loss of generality that $a=2p$ and $b=2p'+1$ for some $p,p'\in\Z$. 
By Corollary~\ref{cor_four}, the two integers that are consecutive with $0$ in $\theta$ and must be odd. Let $q\in\Z$ be the unique integer such that $0$ and $2q+1$ are consecutive, and $0\prec_{\theta} 2q+1$. Observe that $2$ must be the integer that appears immediately after $2q+1$ in $\theta$ (indeed, using again Corollary~\ref{cor_four} we obtain that the number after $2q+1$ in $\theta$ must be even, and since numbers of the same parity appear monotonically increasing it must be number 2). In particular, we have the following relationships in $\theta$:

$$\ldots \prec_{\theta} -2 \prec_{\theta} 0 \prec_{\theta} 2q+1 \prec_{\theta} 2 \prec_{\theta} 4 \prec_{\theta} \ldots$$

Recall that $\theta$ satisfies $\theta=\theta+2$. We apply this equation $|p'-q|$ times and obtain:

$$\ldots \prec_{\theta} -2+2(p'-q) \prec_{\theta} 2(p'-q) \prec_{\theta} 2p'+1=b \prec_{\theta} 2(p'-q)+2 \prec_{\theta} 2(p'-q)+4 \prec_{\theta} \ldots$$

In particular, we have $a\prec_{\theta} b$ if and only if $p\leq p'-q$. The proof for the case in which $2\prec_{\theta} 0$ is identical (the only difference is that we choose $q$ such that $2q+1$ is the number that {\em precedes} $0$ in $\theta$). In either case, this number uniquely determines $\theta$ as either $\alpha_{q}$ or $(\alpha_{q})^{-1}$ as claimed.
\end{proof}

\section*{Proof of Theorem~\ref{theo:dist}}
\begin{proof}
First consider the case in which $\theta \in \{\tau_{o^+e^+},\tau_{e^+o^+}\} \cup \{\alpha_{q} \colon q\in \Z\}$. Let $q=p+(2n,4n,0\ldots,0)$. We claim that  the path from $p$ to $q$ on $\CDS(\theta)$ passes through point $r=p+(2n,2n,0\ldots,0)$.

Note that the slope of $R(p,q)$ is $\mat{t}=(+1, \ldots, +1)$. In particular, by the Translation Lemma (Lemma~\ref{lem:DO_tran}), $R(p,q)$ is a translated copy of the path from the origin to $q-p = (2n,4n,0\ldots,0)$ in $\CDS(\theta-\mat{t}\cdot p)$. Let $\theta' = \theta-\mat{t}\cdot p$ , $q' = (2n,4n,0\ldots,0)$, and $R^{\theta'}(o,q')$ denote the path from the origin to $q'$
in $\CDS(\theta')$. Note that our previous claim is equivalent to saying that $R^{\theta'}(o,q')$ passes through $r' = (2n,2n,0\ldots,0)$.

Indeed, points $o$ and $q'$ share all coordinates except the first two, which implies that $R^{\theta'}(o,q')$ will stay inside the plane $\{x_3=0, x_4=0, \ldots, x_d=0\}$.  Moreover, the $L_1$ distance between $o$ and $q'$ is $6n$. In particular, the segment $R^{\theta'}(o,q')$ will do $6n$ steps, out of which $2n$ will be in the $x_1$-axis and the remaining $4n$ of them in the $x_2$-axis. 

According to the total order construction, we must look at the values of $\theta'$ from $0$ to $\mat{t} \cdot q'-1 =6n-1$ (that is, $\theta'[0,6n-1]$). Since we are moving in the positive quadrant, the $2n$ numbers that are smallest in $\theta'[0,6n-1]$ will correspond to movements in the $x_1$-axis. We claim that all of these $2n$ numbers are smaller than  $\mat{t} \cdot {r}' = 4n$ (in the usual $<$ sense). In particular, all the movements in the $x_1$-axis must be done in the first $4n$ movements from ${o}$. That is, the first $4n$ steps of the segment $R^{\theta'}(o,q')$ contain $2n$ steps in $x_1$-axis and $2n$ steps in $x_2$-axis, so $R^{\theta'}(o,q')$ must pass through ${r}'$ and from there move vertically to $q'$.

For any $\theta \in \{\tau_{o^+e^+},\tau_{e^+o^+}\} \cup \{\alpha_{q} \colon q\in \Z\}$, both odd and even numbers increase monotonically. Since $\theta' = \theta-\mat{t}\cdot p$, both odd and even numbers remain monotonically increasing in $\theta'$. In particular, $\{0 \prec_{\theta'} 2 \prec_{\theta'} 4 \prec_{\theta'} 6 \ldots 6n-2\}$ and $\{1 \prec_{\theta'} 3 \prec_{\theta'} 5 \prec_{\theta'} 7 \ldots 6n-1\} \subset \theta'[0, 6n-1]$. Since the first $2n$ numbers in both sequences are from $0$ to $4n-2$ and $1$ to $4n-1$ respectively, the first $2n$ numbers in any total order containing these two sequences within the interval $[0,6n-1]$ are smaller than $4n$ (in the usual $<$ sense). 

Thus, we conclude that the path from $p$ to $q=p+(2n,4n,0\ldots,0)$ must pass through $r=p+(2n,2n,0\ldots,0)$. %
Using elementary geometry, we can see that the distance from $r$ to $\overline{pq}$ is $\frac{2n}{\sqrt{5}}$, which is a lower bound for $H(\overline{pq},R(p,q))$. 

This completes the proof for the case in which $\theta \in \{\tau_{o^+e^+},\tau_{e^+o^+}\} \cup \{\alpha_{q} \colon q\in \Z\}$. The proof for the remaining case ($\theta \in \{\tau_{o^-e^-},\tau_{e^-o^-}\} \cup \{(\alpha_{q})^{-1} \colon q\in \Z\}$) is very similar. Instead, we look at the path from $p$ to $s=p+(4n,2n,0\ldots,0)$. Using an analogous argument, we can show that $R^{\theta'}(o, s')$ will pass through point $r=p+(2n,2n,0\ldots,0)$. %
\end{proof}
\end{document}